\documentclass[12pt]{article} \usepackage{anysize}
\marginsize{1in}{1in}{1in}{1in}
\usepackage{amsfonts,amssymb,latexsym,amsmath,amscd}
\usepackage{times,mathptmx}

\newcommand{\ketbra}[2]{\ket{#1}\bra{#2}}

\newcommand{\FF}{\mathbb{F}}

\DeclareMathOperator{\SUnitary}{\mathrm{SU}}

\providecommand{\implies}{\Rightarrow}

\newtheorem{theorem}{Theorem}
\newtheorem{definition}[theorem]{Definition}
\newtheorem{proposition}[theorem]{Proposition}

\newtheorem{observation}[theorem]{Observation}
\newtheorem{lemma}[theorem]{Lemma}
\newtheorem{corollary}[theorem]{Corollary}

\newtheorem{theorem*}{Theorem}

\newcommand{\inv}{\Im}

\newenvironment{notation}{\vspace{2mm}\noindent {\bf
    Notation. }}{\vspace{2mm}}

\newenvironment{proof}{{\bf
    Proof.}}{$\blacksquare$ \vspace{2mm}}

\input{Qcircuit}

\newcommand{\controlu}{*-=[][F]{\phantom{\bullet}}}
\newcommand{\multistate}[2]{*+{\hphantom{#2}} \POS[0,0].[#1,0] !C
  *{#2} \POS[0,0].[#1,0] \drop\frm{}}
\newcommand{\ghoststate}[1]{*+{\hphantom{#1}} }
\newcommand{\ccteq}[1]{\multistate{#1}{=}}
\newcommand{\ccteqg}{\ghoststate{=}}
\newcommand{\csl}{{\ \;
    \backslash }}

\newcommand{\qc}{\Qcircuit @C=1em @R=.7em}

\newcommand{\Z}{{\ensuremath{\tt Z}}}
\newcommand{\X}{{\ensuremath{\tt X}}}
\newcommand{\Y}{{\ensuremath{\tt Y}}}
\newcommand{\C}{{\ensuremath{\tt C}}}
\newcommand{\HAD}{{\ensuremath{\tt H}}}
\newcommand{\NOT}{{\ensuremath{\tt NOT}}}
\newcommand{\CNOT}{{\ensuremath{\tt CNOT}}}
\newcommand{\CX}{{\ensuremath{\tt CX}}}
\newcommand{\CCX}{{\ensuremath{\tt CCX}}}
\newcommand{\CZ}{{\ensuremath{\tt CZ}}}
\newcommand{\CCZ}{{\ensuremath{\tt CCZ}}}
\newcommand{\PERES}{{\ensuremath{\tt PERES}}}
\newcommand{\TOFFOLI}{{\ensuremath{\tt TOFFOLI}}}
\newcommand{\MARGOLUS}{{\ensuremath{\tt MARGOLUS}}}

\begin{document}

\title{On the \CNOT -cost of \TOFFOLI\ gates}

\author{
  Vivek V. Shende\thanks{Department of Mathematics,
    Princeton University, Princeton, NJ 08544.} \\
  {\small \tt vshende@princeton.edu}
  \and
  Igor L. Markov\thanks{
    Department of EECS, University of Michigan, Ann Arbor, MI 48109.
  }\\
  {\small \tt imarkov@eecs.umich.edu}
}

\date{}

\maketitle

\abstract{
  The three-input \TOFFOLI\ gate is the workhorse of circuit
  synthesis for classical logic operations on quantum data,
  e.g., reversible arithmetic circuits. In physical implementations,
  however, \TOFFOLI\ gates are decomposed into six \CNOT\ gates and
  several one-qubit gates.  Though this decomposition
  has been known for at least 10 years, we provide here
  the first demonstration of its \CNOT-optimality.

  We study three-qubit circuits which contain less than six
  \CNOT\ gates and implement a block-diagonal operator,
  then show that they implicitly describe the cosine-sine decomposition
  of a related operator.
  Leveraging the canonicity of such decompositions to limit
  one-qubit gates appearing in respective circuits,
  we prove that
   the $n$-qubit analogue of the \TOFFOLI\ requires at least $2n$ \CNOT\ gates.
  Additionally, our results offer a complete classification of three-qubit diagonal
  operators by their \CNOT -cost, which holds even
  if ancilla qubits are available.
}
\newpage
\tableofcontents

\newpage
\section{Introduction}

  The three-qubit
  \TOFFOLI\ gate appears in key quantum logic circuits, such as
  those for modular exponentiation.
  However, in physical implementations it must be decomposed into one-
  and two-qubit gates. Figure \ref{fig:6CNOTs} reproduces the textbook
  circuit from \cite{Nielsen:00} with six \CNOT\ gates,
  as well as Hadamard ($H$), $T = \exp{(i \pi \sigma_z /8)}$
  and $T^\dagger$ gates.

  \begin{figure}[h]
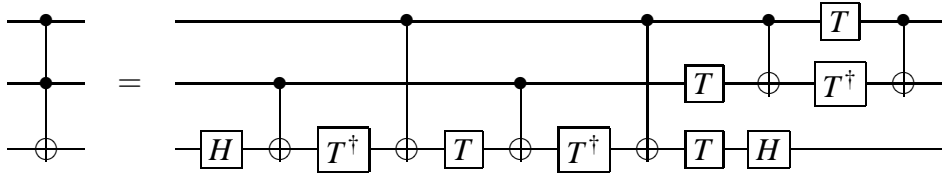

    \begin{center}
      \hspace{0.4cm}
      \qc @C=0.8em {
        & \control \qw & \qw & \ccteq{2} & & \qw & \qw & \qw
        & \control \qw & \qw &
        \qw & \qw & \control \qw & \qw & \control \qw & \gate{T} &
        \control \qw & \qw \\
        & \control \qw \qwx & \qw & \ccteqg & & \qw & \control \qw & \qw &
        \qw \qwx & \qw & \control \qw & \qw & \qw \qwx &
        \gate{T} & \targ \qwx & \gate{T^\dagger} & \targ \qwx  & \qw \\
        & \targ \qw \qwx & \qw & \ccteqg & & \gate{H} & \targ \qwx &
        \gate{T^\dagger} & \targ \qwx & \gate{T} & \targ \qwx &
        \gate{T^\dagger} & \targ \qwx & \gate{T} & \gate{H} & \qw & \qw &
        \qw }
      \caption{\label{fig:6CNOTs}
        Decomposing the \TOFFOLI\ gate into one-qubit and six \CNOT\
        gates.}
    \end{center}
  \end{figure}

  The pursuit of efficient circuits for standard
  gates has a long and rich history.
  DiVincenzo and Smolin found numerical evidence \cite{DiVincenzo:94}
  that five two-qubit gates are
  necessary and sufficient to implement the \TOFFOLI.
  Margolus showed that a phase-modified
  \TOFFOLI\ gate admits a three-\CNOT\ implementation
  \cite{Margolus:94,DiVincenzo:98},
  whose optimality was eventually demonstrated by Song and
  Klappenecker \cite{Song:margolus:03}.
  Unfortunately, this \MARGOLUS\ gate can
  replace \TOFFOLI\ only in rare cases.
  The detailed case analysis used in the optimality proof
  from \cite{Song:margolus:03} does not extend easily
  to circuits with four or five \CNOT s.
  The omnibus Barenco et al. paper offers circuits for many standard
  gates, including an eight-\CNOT\ circuit
  for the \TOFFOLI\
  \cite[Corollary 6.2]{Barenco:elementary:95}, as well as
  a six-\CNOT\ circuit for the controlled-controlled-$\sigma_z$,
  which differs from the \TOFFOLI\ only by one-qubit
  operators \cite[Section 7]{Barenco:elementary:95}.
  Problem 4.4b of the textbook by Nielsen and Chuang asks whether
  the circuit of Figure \ref{fig:6CNOTs}
  could be improved. The problem was marked as unsolved, and we report
  the following progress.

  \begin{theorem}
    A circuit consisting of \CNOT\ gates and one-qubit gates which implements
    the $n$-qubit \TOFFOLI\ gate without ancillae
    requires at least $2n$ \CNOT\ gates.
    For $n=3$, this bound holds even when ancillae are permitted, and
    is achieved by the circuit of Figure \ref{fig:6CNOTs}.
  \end{theorem}

  Our main tool is the Cartan decomposition in its ``KAK'' form, which
  provides a Lie-theoretic generalization of the
  singular-value decomposition \cite{Knapp:98}.
  Several special cases have previously proven useful for
  the synthesis and analysis of quantum circuits, notably
  the two-qubit {\em magic decomposition}
  \cite{Lewenstein:00,Makhlin:00,Zhang:2qgeom:02,Vidal:2q:04,Vatan:2q:04,Shende:2q:04,Shende:cnotcount:04},
  {\em the cosine-sine decomposition}
  \cite{Khaneja:01, Bullock:kgd, Mottonen:csd:04,
    Shende:qsd:05}, and {\em the demultiplexing decomposition}
  \cite{Shende:qsd:05}.
  The canonicity of the two-qubit canonical decomposition was used
  previously to perform \CNOT-counting for two-qubit operators
  \cite{Shende:2q:04}.
  The magic decomposition is a two-qubit phenomenon,\footnote{
    While the Cartan decomposition
    $\mathrm{SU}(n) = \mathrm{SO}(n)\cdot \mathrm{[diagonals]}\cdot \mathrm{SO}(n)$
    is general, the utility of the magic decomposition
    arises from the isomorphism
    $\mathrm{SU}(2) \times \mathrm{SU}(2) \simeq \mathrm{SO}(4)$
    being represented
    as an inner automorphism of $\mathrm{SU}(4)$.
    Such coincidental isomorphisms are few and confined
    to low dimensions.
  }
  but the cosine-sine and demultiplexing decompositions hold
  for $n$-qubit operators
  and enjoy similar canonicity.
  Moreover, the components of these decompositions are {\em multiplexors}
  \cite{Shende:qsd:05} --- block-diagonal operators
  that commute with many common circuit elements.
  Commutation properties facilitate circuit restructuring that can dramatically
  reduce the number of circuit topologies to be considered in proofs.
  These results and observations allow us to perform \CNOT-counting
  using the Cartan decomposition in a divide-and-conquer manner.

  In the remaining part of this paper, we first review basic
  properties of quantum gates in Section \ref{sec:basic} and make
  several elementary simplifications to reduce the complexity
  of the subsequent case analysis.
  In particular, we pass from the \CNOT\ and \TOFFOLI\
  gates to the symmetric, diagonal \CZ\ and \CCZ\ gates, and
  recall circuit decompositions which yield operators
  commuting with \Z\ and \CZ\ gates.  We also define
  qubit-local \CZ-costs, and observe that the total \CZ-cost
  can be lower-bounded by half the sum of
  the local \CZ\ counts for each qubit.  Though weak,
  this bound suffices for our purposes and we can
  compute it in simple cases. Further technique is developped in
    Section \ref{sec:snk}, where we compute matrix
    entries to derive
  constraints on gates from circuit equations.  This
    approach was employed by Song and
    Klappenecker in the two-qubit case,
    and we generalize several of their
    results to $n$-qubit circuits.

  Section \ref{sec:toffoli} is the heart of the present work,
  in which we prove our result on the \CNOT-cost of the \TOFFOLI\ gate.
  It starts by motivating and outlining the methods involved,
  previews key intermediate results, and proves that the \CNOT-cost
  of the \TOFFOLI\ is 6, based on these results. In Section
  \ref{sec:avsc}, we use the canonicity of the cosine-sine
  decomposition derive circuit constraints.
  Section \ref{sec:invariants}, motivated by \cite{Shende:cnotcount:04},
  employs the canonicity of the demultiplexing decomposition, captured
  by a spectral invariant, to lower-bound \CZ\ gates required in circuit
  implementations of operators.
  The results apply, {\em mutatis mutandis}, to \CNOT-based
  implementations as well.  Finally, in Section \ref{sec:corollaries},
  we deduce as corollaries that
  the three-qubit \PERES\ gate requires exactly 5 \CNOT s and the
  $n$-qubit \TOFFOLI\ gate requires at least $2n$.
  In Section \ref{sec:3diag}, we extend our techniques to all
  three-qubit diagonal operators,
  completely classifying them according to \CZ -cost.
  Generalizations to circuits with ancillae are obtained in
  Section \ref{sec:ancilla}.
  Concluding discussion can be found
  in Section \ref{sec:conclusion}.

  \section{Preliminaries}
  \label{sec:basic}

  We review notation and properties of useful quantum gates,
  then characterize operators that commute with Pauli-\Z\ gates
  on multiple qubits.  We then review circuit decompositions from
  \cite{Bullock:diag:04, Mottonen:csd:04, Shende:qsd:05}.
  Finally, we introduce terminology appropriate for quantifying
  gate costs of unitary operators
  in terms of the \CNOT\ and \CZ\, and state
  elementary but useful observations about these costs.

  \subsection{Notation and properties of standard quantum gates}

  We write $\X, \Y, \Z$ for the Pauli operators, and
  $\CX, \CCX$ for $\CNOT, \TOFFOLI$. Rotation gates  $\exp(i \Z \theta)$
  are denoted by $R_z(\theta)$, and we analogously use $R_x, R_y$.\footnote{
  We omit the factor of $\pm 1/2$ used by other authors.}
  We work throughout on some fixed number of qubits $N$.
  For a one-qubit gate $g$ and a qubit $q$,
  we denote by $g^{(q)}$ the $N$-qubit operator implemented
  by applying the gate $g$ on qubit $q$.  Similarly,
  $\C^{(i)}\X^{(j)}$ is the operator implemented by a
  controlled-$\X$ with the control on qubit $i$ and target on qubit $j$.
  The controlled-$\Z$ being symmetric with respect to exchanging qubits,
  we do not distinguish control from target in the notation $\CZ^{(i,j)}$.
  We similarly denote the operator of
  a controlled-controlled-\Z\ on qubits $i,j,k$
  by $\CCZ^{(i,j,k)}$.
  {\em In choosing qubit labels, we follow throughout the convention that
  the high-to-low significance order of qubits is the same as the
  lexicographic order of their labels.}

  We follow the standard but sometimes confusing
  convention that {\em typeset operators act on vectors from the left},
  but {\em circuit diagrams process inputs from the right.}
  Consistently with the established notation for the \CNOT\ gate,
  we denote the $\X$ gate by ``$\oplus$'' in circuit diagrams.
  We denote the $\Z$ gate by a ``$\bullet$'' symbol, which does not
  lead to ambiguity in the matching notation for \CZ\ because
  \CZ\ is symmetric.
  Thus the following diagram expresses the identity
  $\CZ^{(\ell,m)} \X^{(\ell)} =  \Z^{(m)} \X^{(\ell)}
  \CZ^{(\ell,m)}$ and rearranges gates in quantum circuits,
  like de Morgan's law does in digital logic.

  \begin{equation}
    \label{circ:demorgan}
    \qc{
      \lstick{\ell} & \targ \qw & \control \qw & \qw
      & \ccteq{1} & & \control \qw & \targ & \qw \\
      \lstick{m}    &       \qw & \control \qw \qwx & \qw
      & \ccteqg   & &  \control \qw \qwx & \control \qw & \qw
    }
  \end{equation}

  Another standard identity relates the \X , \Z , and
  one-qubit
  ${\tt HADAMARD}$ (\HAD) gates:   $\HAD\X\HAD = \Z$.
  By case analysis on control qubits, one obtains the further identities
  $\HAD^{(i)} \C^{(j)}X^{(i)} \HAD^{(i)} = \CZ^{(i,j)}$
  and $\HAD^{(i)} \C \C^{(j,k)} \X^{(i)} \HAD^{(i)} = \CCZ^{(i,j,k)}$.
  Despite this equivalence, we prefer the \X\ family of gates
  for some applications
  and the \Z\ family for others, as summarized in Table
  \ref{tab:proscons}.

  Circuits consisting entirely of one-qubit gates and \CZ\ (respectively
  \CNOT ) gates will be called \CZ-circuits (respectively \CNOT-circuits).
  Using the above identities, \CZ-circuits and \CNOT-circuits
  can be interchanged at the cost of adding one-qubit \HAD\ gates.
  It will also be convenient to consider $\CZ^{(\ell)}$-circuits, which
  by definition are arbitrary circuits where all multi-qubit gates touching
  qubit $\ell$ are $\CZ$. While these are not a subclass of $\CZ$-circuits,
  a $\CZ^{(\ell)}$-circuit can be converted into a $\CZ$-circuit
  without any changes affecting qubit $\ell$.

  \begin{table}
    \begin{center}
{     \small
      \begin{tabular}{l|l|l}
        & \CNOT\ and \TOFFOLI\ & \CZ\ and \CCZ\ \\
        \hline
        Advantages & \multicolumn{2}{c}{With one-qubit gates added, either
          \CNOT\ or \CZ\
          would be universal} \\
        \cline{2-3}
        & Implement addition and multiplication & Symmetric \\
        & Universal for reversible computation & Fewer circuit topologies\\
        & Block-diagonal & Diagonal \\
        & With 1-qubit diagonals, implement any diagonal & ---    \\
        & Commute with \X\ on target & Commute with \Z\ on target\\
        \hline
        Other      & Change direction after two \HAD-conjugations & \\
        \cline{2-3}
        \ \ properties & \multicolumn{2}{c}{One can map back and forth by
          \HAD-conjugation on target} \\
        \hline
        Applications & \bf Circuit synthesis & \bf Circuit analysis \\
        \hline
      \end{tabular}
}
      \caption{
        \label{tab:proscons}
        Relative advantages of standard controlled gates.
      }
    \end{center}
  \end{table}

 \subsection{Operators commuting with \Z}
  We now recall terminology for operators commuting with \Z\
  on some qubits, but possibly not all qubits.
  Further background on the circuit theory of these
  {\em quantum multiplexors} can be found in \cite{Shende:qsd:05}.

  The control-on-box notation of the following diagram indicates
  that the operator $U$ commutes with $\Z^{(\ell)}$.
  The backslash on the bottom
  line indicates an arbitrary number of qubits (a multi-qubit bus).

  \[
  \qc{
    \lstick{\ell} & \controlu \qw & \qw \\
    \csl & \gate{U} \qwx & \qw
  }
  \]

  \noindent
  These operators include the commonly-used positively
  and negatively controlled-$U$ gates, although in our
  notation $U$ also acts on the control qubits (and is thus
  ``larger than the box in which it is contained'').
  In general, operators which commute with \Z\ are block-diagonal:

  \begin{observation}
    For a unitary operator $Q$ and qubit $\ell$, consider
    the one-qubit values $\ket{0}^{(\ell)}$ and $\ket{1}^{(\ell)}$
    on $\ell$-th input and output qubits of the operator.
   The following are equivalent.
    \begin{itemize}
      \item $Q$ commutes with $\Z^{(\ell)}$
      \item $\bra{0}^{(\ell)} Q \ket{1}^{(\ell)} = 0$
      \item $\bra{1}^{(\ell)} Q \ket{0}^{(\ell)} = 0$
      \item $Q$ admits a decomposition
        $Q  = \ket{0} \bra{0} \otimes Q_0 + \ket{1} \bra{1} \otimes Q_1$, where the
        projectors $\ket{i} \bra{i}$ operate on qubit $\ell$ and the unitary $Q_i$
        operate on the qubits other than $\ell$.
    \end{itemize}
  \end{observation}
  In an appropriate basis, the matrix of $Q$ is block-diagonal.
  Its blocks represent the ``then'' and
  ``else'' branches of the quantum multiplexor $Q$ with select
  qubit $\ell$.

  \begin{notation}
    If $Q$ commutes with $\Z^{(\ell)}$ and $\ell$ is clear from
    context, we denote $Q$'s diagonal blocks $\bra{j}^{(\ell)} Q
    \ket{j}^{(\ell)}$ by $Q_j$.
    Similarly, if $Q$ commutes with with $\Z^{(\ell_i)}$ on multiple qubits
    $\ell_1 \ldots \ell_k$, then for any bitstring $j_1 \ldots j_k$
    we write $Q_{j_1 \ldots j_k}$ for
    $\bra{j_1 \ldots j_k}^{(\ell_i \ldots \ell_k)} Q
    \ket{j_1 \ldots j_k}^{(\ell_i \ldots \ell_k)}$.
    \end{notation}

      When the $\ell_k$ include all the qubits, $Q$ is diagonal
    and the $Q_{j_1 \ldots j_k}$ are its diagonal entries.
    In general, $Q_{j_1 \ldots j_k}$ capture diagonal blocks
    of $Q$ with respect to an ordering of computational-basis vectors
    in which qubits $\ell_1 \ldots \ell_k$ are the
    most significant qubits.

  We now point out the following commutability.

  \begin{observation}
    Let $Q, R$ be two gates such that
    for every qubit $\ell$, either one of them does not affect $\ell$,
    or both of them commute with $\Z^{(\ell)}$.  Then $QR = RQ$. In picture:
  \end{observation}
    \[
    \qc{ \csl & \controlu \qw & \controlu \qw & \qw & \ccteq{2} &
      \csl & \controlu \qw & \controlu \qw & \qw \\
      \csl & \gate{Q} \qwx & \qw \qwx & \qw & \ccteqg &
      \csl & \qw \qwx      & \gate{Q} \qwx & \qw \\
      \csl & \qw & \gate{R} \qwx & \qw & \ccteqg & \csl & \gate{R}
      \qwx & \qw & \qw }
    \]

  We now recall the {\em multiplexed rotation}
  gates \cite{Mottonen:csd:04, Shende:qsd:05}, which generalize
  the $R_x, R_y, R_z$ gates.  Let $\Delta$ be a diagonal Hermitian matrix
  acting on the qubits $\ell_1, \ldots ,\ell_k$, and fix another
  qubit $m \ne \ell_i$.  We define the operator $R_z^{(m)}(\Delta)$
  on the qubits $\ell_1, \ldots, \ell_k,m$ by the conditions
  (1) that it commute with $\Z^{(\ell_i)}$ for all $i$, and
  (2) for any bitstring $j_1 \ldots j_k$, we have
  $R_z^{(m)}(\Delta)_{j_1 \ldots j_k} = R_z(\Delta_{\ell_1 \ldots \ell_k})$.
  Explicitly,
  $R_z^{(m)}(\Delta) = \exp(i \Z^{(m)} \Delta^{(\ell_1 \ldots \ell_k)})$.
  Multiplexed $R_x, R_y$ gates are
  defined similarly.  Since such operators commute with
  $\Z^{(\ell_i)}$, we depict them in circuit diagrams
  with the appropriate control-on-boxes.

  It is natural to ask when an operator commuting with various $\Z$
  gates can be implemented in a \CZ-circuit containing only
  gates commuting with the same $\Z$ gates.  The answer
  is given in terms of the {\em partial determinant}.

  \begin{definition}
    Fix qubits $\ell_1 \ldots \ell_k$.
    We define the partial determinant map
    $\det_{\ell_1 \ldots \ell_k}$ from the operators
    commuting with $\Z^{(\ell_1)}, \ldots, \Z^{(\ell_k)}$
    to the diagonal operators acting only on the qubits $\ell_i$.
    It is given by $(\det_{\ell_1 \ldots \ell_k}(U))_{j_1 \ldots j_k} =
    \det (U_{j_1 \ldots j_k})$.
  \end{definition}

  When computing partial determinants of a single gate or subcircuit acting on $m$ qubits,
  we first tensor respective operators with $I_{2^{N-m}}$ to form operators
  acting on all $N$ qubits (which may affect the determinants).
  When applied to such ``full'' operators, the partial determinant mapping
  is a group homomorphism.

  \begin{proposition} \label{prop:noninvariance}
    Fix qubits $\ell_1 \ldots \ell_k$ among $N > k$ qubits.
    A unitary $U$ commuting with $\Z^{(\ell_1)}, \ldots,
    \Z^{(\ell_k)}$
    can be implemented by a \CZ-circuit in which only diagonal
    gates operate on qubits $\ell_i$ if and only if
    $\det_{\ell_1 \ldots \ell_k}(U)$
    is separable (can be implemented by one-qubit gates).
  \end{proposition}
  \begin{proof}
    ($\Rightarrow$).
    It suffices to show the separability of
    $\det_{\ell_1 \ldots \ell_k}(U)$ for a generating set
    of operators.  By definition, such a generating set is provided by \CZ s,
    one-qubit diagonals on the $\ell_i$, and gates not affecting any of the $\ell_i$.

    Note first that any diagonal gate $D$ acting on qubits
    $\ell_1, \ldots, \ell_k$ has partial determinant given by
    $\det_{\ell_1 \ldots \ell_k}(D) = D^{2^{N-k}}$, understood as an operator
    on qubits $\ell_1 \ldots \ell_k$.  In particular, if $D$ were
    separable, then so is $\det_{\ell_1 \ldots \ell_k}(D)$.  If
    $D = \CZ^{(\ell_i, \ell_j)}$, then from $\CZ^2 = I$ and
    $N > k$ we deduce $\det_{\ell_1 \ldots \ell_k}(\CZ^{(\ell_i, \ell_j)}) = I$.
    The remaining gates we need to consider are:

    \noindent {\bf (i)} any gate not affecting qubits $\ell_i$
    implements $U= Q^{(1..N)\setminus(\ell_1 \ldots \ell_k)}$
    for some $Q$. \\
    In this case $U_{j_1 \ldots j_k}=Q$, and furthermore
    $\det_{\ell_1 \ldots \ell_k}(U)=\det(Q)I$.

    \noindent {\bf (ii)} \CZ\ gates connecting qubits $\ell_i, m \notin
    \{\ell_1, \ldots, \ell_k\}$.  We compute
    $\det_{\ell_1 \ldots \ell_k}(\CZ^{(\ell_i, m)}) =
     (\Z^{(\ell_i)})^{2^{N-k - 1}}$.

    ($\Leftarrow$). This part of the result is not used in the rest
    of the paper, and we therefore defer the proof to the
    Appendix.
  \end{proof}

  \subsection{Cartan decompositions in quantum logic}
  \label{sec:cartan}

  This section recalls two important operator decompositions
  ({\em cosine-sine} and {\em demultiplexing}) and casts them
  as circuit decompositions. Readers willing to accept their use
  in our proofs may skip to Section \ref{sec:CZcounting}.

  Observe that
  an operator can be implemented with a single one-qubit
  gate if and only if it commutes with the Pauli operators
  \Z\ and \X\ on all other
  qubits.  Thus to produce a \CNOT- circuit for a given
  operator $U$, one may use the following algorithmic framework.
  \begin{enumerate}
  \item Decompose $U$ into a circuit in which each
    non-\CNOT\ gate, $V, W, \ldots$, commutes with \X\ and \Z\ on more
    qubits that $U$ does.
  \item Apply the algorithm recursively to $V, W, \ldots$
        until one-qubit gates are reached.
  \end{enumerate}

  As \Z\ is self-adjoint, the
  requirement that $U$ commutes with $\Z^{(i)}$ can be rephrased
  as the condition that $U$ is fixed under the involution
  $U \mapsto \Z^{(i)} U \Z^{(i)}$.  Given such an involution,
  a fundamental Lie-theoretic result produces an operator decomposition
  \cite{Knapp:98}.
  Here we recite the result for completeness, but do not require
  the reader to understand all terminology.

  \vspace{4mm} \noindent {\bf The Cartan Decomposition.} Let $G$ be a
  reductive Lie group, and $\iota:G \to G$ an
  involution. Let $K = \{g:\iota(g)=g\}$ and $A$ be maximal over
  subgroups contained in $\{g:\iota(g)=g^{-1}\}$.
  Then $K$ is reductive, $A$ is
  abelian, and $G = KAK$. \vspace{4mm}

  In order to restate decompositions of unitary operators
  as circuit decompositions, we employ the notation of
  {\em set-valued} quantum gates \cite{Shende:qsd:05}.  Completely
  unlabelled gates (as in Equation \ref{eq:demux:diag}) denote the set
  of all gates satisfying all control-on-box commutativity conditions
  imposed by the diagram, and gates labelled $R_x, R_y, R_z$ denote
  the appropriate set of (possibly multiplexed) rotations.
  An equivalence of circuits with
  set-valued gates means that if we pick an element from each set on
  one side, there is a way to choose elements on the other so that the
  two circuits compute the same operator.  The backslashed wires
  which usually indicate multiple qubits may also carry {\em zero} qubits.

  The involution $\phi_Z: U \mapsto \Z^{(\ell)} U \Z^{(\ell)}$ corresponds
  to the {\em cosine-sine decomposition}.\footnote{The terminology
    comes from the numerical linear algebra literature; see
    \cite{Paige:94} and references therein.}

  \begin{equation} \label{eq:csd}
    \qc {
      \csl & \controlu \qw & \qw & \ccteq{2} & & \controlu \qw & \controlu \qw
      & \controlu \qw & \qw \\
      \csl & \multigate{1}{\phantom{U}} \qwx & \qw & \ccteqg & &
      \gate{\phantom{U}} \qwx & \controlu \qw \qwx & \gate{\phantom{U}} \qwx
      & \qw \\
      \lstick{\ell} & \ghost{U} & \qw & \ccteqg & & \controlu \qw \qwx
      & \gate{R_y} \qwx & \controlu \qw \qwx & \qw \\
    }
  \end{equation}

  The involution $\phi_Y: U \mapsto \Y^{(\ell)} U \Y^{(\ell)}$ yields
  the {\em demultiplexing decomposition} \cite{Shende:qsd:05}.

  \begin{equation}    \label{eq:demux}
    \qc{
      \lstick{\ell} & \controlu \qw & \qw & \ccteq{2} &
      & \qw & \gate{R_z} & \qw & \qw \\
      \csl & \controlu \qw \qwx & \qw & \ccteqg &
      & \controlu \qw  & \controlu \qw \qwx  & \controlu \qw
      & \qw \\
      \csl & \gate{\phantom{U}} \qwx
      & \qw & \ccteqg   & & \gate{\phantom{U}} \qwx &
      \controlu \qw \qwx & \gate{\phantom{U}} \qwx & \qw \\
    }
  \end{equation}

  The map $\phi_Y$ restricts to the subgroup of
  diagonal operators.  This group being
  abelian, the $K$ and $A$ factors commute, leaving
  the following decomposition of diagonal operators.

  \begin{equation} \label{eq:demux:diag}
    \qc {
      \lstick{\ell} & \controlu \qw & \qw & \ccteq{1} &
      & \qw & \gate{R_z} & \qw \\
      \csl & \controlu \qw \qwx & \qw & \ccteqg &
      & \controlu \qw & \controlu \qwx \qw & \qw \\
    }
  \end{equation}

  The involution $\phi_Y$ further restricts to the
  subgroup of multiplexed $\Z$ rotations, which
  we can demultiplex again.
  The $K$ and $A$ factors again commute; the $A$ factor
  is computed by the last 3 gates in the circuit below.

  \begin{equation} \label{eq:demux:rz}
    \qc{
      \lstick{\ell} & \controlu \qw & \qw & \ccteq{2} &
      & \qw & \control \qw & \qw & \control \qw
      & \qw \\
      \csl & \controlu \qw \qwx & \qw & \ccteqg & &
      \controlu \qw & \qw \qwx & \controlu \qw & \qw \qwx &
      \qw \\
      & \gate{R_z} \qwx{2} \qw & \qw & \ccteqg & &
      \gate{R_z} \qwx & \targ \qwx{2} &
      \gate{R_z} \qwx & \targ \qwx{2} & \qw
    }
  \end{equation}

  To establish the existence of these decompositions, it remains
  to verify in each case that the purported $K$ and $A$ satisfy
  the appropriate properties with respect to the relevant
  involution.  This can be checked after passing to the
  Lie algebra where it is easy.
  Alternatively, explicit constructions of the cosine-sine and demultiplexing
  decompositions are given in \cite{Paige:94} and \cite{Shende:qsd:05},
  respectively.

  To decompose general $n$-qubit operators, Equation \ref{eq:csd}
  can be applied iteratively until all remaining gates are either
  multiplexed $R_y$ gates or diagonal.  The $R_y$ gates
  can be replaced by $R_z$ gates at the cost of introducing
  some one-qubit operators; the $R_z$ and other diagonal gates
  can be decomposed as described above; for details
  and optimizations see \cite{Mottonen:csd:04}.
  Smaller circuits are obtained by another algorithm,
  which alternates cosine-sine decompositions with
  demultiplexing decompositions; for details
  and optimizations, see \cite{Shende:qsd:05}.

  When circuit decompositions are applied recursively,
  some gates can be reduced by local circuit transformations.
  For example, when iteratively demultiplexing multiplexed $R_z$ gates,
  some \CNOT s may be cancelled as shown below.

  \[
    \qc {
       & \qw & \ctrl{3} \qw & \qw &
       \ctrl{3} \qw & \qw & \ccteq{3} & & \qw & \qw & \qw & \qw & \ctrl{3}
       \qw & \qw & \qw & \qw & \qw & \ctrl{3} \qw & \qw \\
       & \controlu \qw & \qw &
       \controlu \qw & \qw & \qw & \ccteqg  & & \qw & \ctrl{2} \qw & \qw
       & \ctrl{2} \qw & \qw & \ctrl{2} \qw & \qw & \ctrl{2} \qw & \qw
       & \qw & \qw \\
       \csl & \controlu \qw \qwx & \qw &
       \controlu \qw \qwx & \qw & \qw & \ccteqg & & \controlu \qw & \qw &
       \controlu \qw & \qw & \qw & \qw & \controlu \qw & \qw &
       \controlu \qw & \qw & \qw \\
       & \gate{R_z} \qw \qwx & \targ &
       \gate{R_z} \qw \qwx & \targ & \qw & \ccteqg & & \gate{R_z} \qw \qwx
       & \targ &
       \gate{R_z} \qw \qwx & \targ & \targ & \targ & \gate{R_z} \qw \qwx
       & \targ &
       \gate{R_z} \qw \qwx & \targ & \qw \gategroup{2}{12}{4}{12}{1em}{--}
       \gategroup{2}{14}{4}{14}{1em}{--}
    }
  \]

  This technique produces a circuit with $2^n$ \CNOT\ gates for
  an $n$-ply multiplexed $R_z$ gate.  Using Equation \ref{eq:demux:diag},
  we obtain a circuit with $2^n - 2$ \CNOT\ gates for an arbitrary
  $n$-qubit diagonal operator \cite{Bullock:diag:04}.
  Applying this result to \CCZ\ gate leads to the circuit in Figure
  \ref{fig:6CNOTs}.

  \subsection{Basic facts about \CZ-counting}
  \label{sec:CZcounting}

  The \CZ-cost $|U|_{\CZ}$ of an $N$-qubit operator $U$
  is the minimum number of \CZ s which appear
  in any $N$-qubit \CZ-circuit for $U$; we define
  the \CNOT-cost analogously.  The identity
  $\HAD^{(i)} \C^{(j)}X^{(i)} \HAD^{(i)} = \CZ^{(i,j)}$
  ensures that $|U|_{\CZ} = |U|_{\CNOT}$.
  The further identity
  $\HAD^{(i)} \C \C^{(j,k)} \X^{(i)} \HAD^{(i)} = \CCZ^{(i,j,k)}$
  yields:

  \begin{observation} $|\CCZ|_{\CZ} = |\CCX|_{\CNOT} \le 6$.
  \end{observation}

  By way of illustration, the following modification of the
  circuit in Figure \ref{fig:6CNOTs} implements the \CCZ\ in terms of \CZ s.

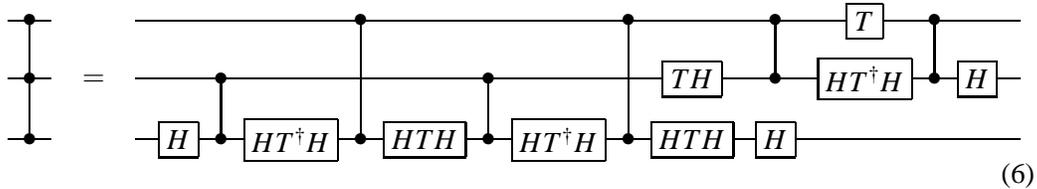
  \begin{equation}
    \label{circ:ccz}
    \qc @C=0.8em{
      &
      \control \qw & \qw & \ccteq{2} & & \qw & \qw & \qw & \control \qw & \qw &
      \qw & \qw & \control \qw & \qw & \control \qw & \gate{T} &
      \control \qw & \qw & \qw \\
      &
      \control \qw \qwx & \qw & \ccteqg & & \qw & \control \qw & \qw &
      \qw \qwx & \qw & \control \qw & \qw & \qw \qwx &
      \gate{TH} & \control \qw \qwx & \gate{HT^\dagger H} & \control \qw \qwx
      & \gate{H} & \qw \\
      &
      \control \qw \qwx & \qw & \ccteqg & & \gate{H} & \control \qw \qwx &
      \gate{HT^\dagger H} & \control \qw \qwx & \gate{HTH} & \control \qw \qwx &
      \gate{HT^\dagger H} & \control \qw \qwx & \gate{HTH} & \gate{H}
      & \qw & \qw &
      \qw & \qw }
  \end{equation}

  It shall prove more convenient to compute $|\CCZ|_{\CZ}$ rather than
  $|\CCZ|_{\CNOT}$.
  To do so, we are going to study the number of \CZ s which must touch
  a given qubit in any \CZ-circuit for a given operator.
  More precisely,  the
  $\CZ^{(\ell)}$-cost $|U|_{\CZ;\ell}$ is the minimum
  number of \CZ\ gates incident on $\ell$ in any $\CZ^{(\ell)}$-circuit
  for $U$. These cost functions are related through the following
  estimate.\footnote{
    This bound is very weak in general.  Dimension-counting
  shows that a generic $N$-qubit operator $U$ requires
  on the order of $4^N$ \CZ\ gates \cite{Knill:95}, whereas
  the results of \cite{Shende:qsd:05} imply that $|U|_{\CZ;\ell} < 6N$.
  At best we can establish that $|U|_{\CZ} \ge N(6N-1)$.}
  \begin{observation}
    For any operator $P$,
    \[|P|_{\CZ} \ge \frac{1}{2} \sum_{j} |P|_{\CZ;j}\]
  \end{observation}
  \begin{proof}
    Each $\CZ$ gate touches two qubits.
  \end{proof}

  As the costs $|\CCZ|_{\CZ;j}$
  are the same for $j=1,2,3$ (by symmetry),
  \begin{equation} \label{eq:threehalfs}
    |\CCZ|_{\CZ} \ge \frac{3}{2}|\CCZ|_{\CZ;j}
  \end{equation}

  We emphasize that the number of qubits, $N$,
  is an unspecified parameter in both $|\cdot|_{\CZ}$ and
  $|\cdot|_{\CZ;\ell}$.  In the presence of ancillae,
  we define
  $|U|^a_{\CZ} : = \min_t |U \otimes I_2^{\otimes t}|_{\CZ}$.
  Obviously $|U|^a_{\CZ} \le |U|_{\CZ}$.  While
  $|U|^a_{\CZ} = |U|_{\CZ}$ seems unlikely to always hold, we
  are not aware of any counterexamples.  Indeed, we
  will show in Section  \ref{sec:ancilla} that this equality
  holds for all two-qubit operators and all three-qubit
  diagonal operators.

  \section{Deriving gate constraints from circuit equations}
  \label{sec:snk}

   The circuit decompositions of Section \ref{sec:cartan}
   are essentially unique, and from this canonicity
   one can derive various constraints on which gates
   may appear in certain circuit equations.  We will
   pursue this route in Section \ref{sec:avsc}.
   However, the simplest cases are easier to treat
   from the more elementary point of view adopted
   by Song and Klappenecker in their classification
   of two-qubit controlled-$U$ operators by \CNOT-cost
   \cite{Song:2q:02}. Considering the operator
   computed by a candidate circuit, they first focus
   on matrix elements which vanish if the operator
   is a controlled-$U$. In order to produce such zero elements,
   the gates in the candidate circuit must satisfy certain constraints.
   Below we derive a series of more general results for $n$-qubit circuits.
   One-qubit gates which become diagonal when multiplied by $\X$ occur
   frequently; we refer to them as anti-diagonal.

  \begin{lemma} \label{lem:sp1}
    The following equation imposes at least one of the following constraints.
    \[
    \qc{
      \lstick{1} & \gate{b}    & \controlu \qw & \gate{a} & \qw
      & \ccteq{1} & & \controlu \qw & \qw \\
      \csl & \qw & \gate{P} \qwx & \qw & \qw & \ccteqg & & \gate{Q} \qwx & \qw
    }
    \]
    \begin{enumerate}
    \item  $a,b$ are both diagonal or both anti-diagonal.
    \item  $P$ takes the form $d \otimes P_0$ for some
      one-qubit diagonal $d$.
    \end{enumerate}
  \end{lemma}
  \begin{proof}
    \[0 = \bra{0}^{(1)} aPb \ket{1}^{(1)} = \bra{0}
    a\ket{0}\bra{0}b \ket{1} P_0 + \bra{0}a\ket{1}\bra{1}b \ket{1}
    P_1\] As the coefficients do not vanish, $P_0$ and $P_1$ are
    linearly dependent.  It follows that $P = d \otimes P_0$ for
    some one-qubit diagonal $d$.
  \end{proof}

  \begin{corollary} \label{cor:sp1:CZ}
    If $a^{(i)} \CZ^{(i,j)} b^{(i)}$ commutes with $\Z^{(i)}$, then
    $a,b$ are both diagonal or anti-diagonal.
  \end{corollary}

  \begin{corollary} \label{cor:sp1:replace}
    In the situation of Lemma \ref{lem:sp1}, there exist
    one-qubit operators
    $a',b'$ which are either diagonal or anti-diagonal, such that
    $a'^{(1)}Pb'^{(1)}=Q$.
  \end{corollary}
  \begin{proof}
    Apply Lemma \ref{lem:sp1}; we need consider only
    Case 2.  Take $a' = a \delta b \delta^{-1}$ and
    $b' = I$; then $a'^{(1)}Pb'^{(1)} = a^{(1)}Pb^{(1)}$.  As
    $a'^{(1)} =QP^\dagger$
    commutes with $\Z^{(1)}$, it is diagonal.
  \end{proof}

  We turn now to circuits with two \CZ\ gates.

  \begin{lemma} \label{lem:sp2}
    Suppose the following equation holds.
    \[
    \qc{
      \lstick{1} & \gate{b}    & \qw & \gate{a} & \qw
      & \ccteq{2} & & \controlu \qw & \qw \\
      \lstick{2} &  \controlu \qw \qwx & \multigate{1}{P} &
      \controlu \qw \qwx & \qw
      & \ccteqg & & \multigate{1}{Q} \qwx & \qw \\
      \csl & \qw & \ghost{P} & \qw & \qw & \ccteqg & & \ghost{Q} & \qw}
    \]
    Then (I) $a_i b_j$ is diagonal for all $i,j$ or
    (II) one of $P$, $\X^{(2)}P$ commutes with $\Z^{(2)}$.
  \end{lemma}
  \begin{proof}
    We compute:
    \[0 = \bra{0}^{(1)} \bra{i}^{(2)} aPb \ket{1}^{(1)}
    \ket{j}^{(2)} = \bra{0}^{(1)}a_i b_j \ket{1}^{(1)} \bra{i}^{(2)} P
    \ket{j}^{(2)}\] Either $\bra{i}^{(2)} P \ket{j}^{(2)} = 0$ for
    some $i,j$, or $\bra{0} a_i b_j \ket{1}$ vanishes for all $i,j$.
  \end{proof}

  \begin{corollary}
    Suppose the following equation holds.
    \[
    \qc{
      \lstick{1}
      & \controlu \qw & \qw & \ccteq{2} & & \gate{t} & \control \qw
      & \gate{s}         & \control \qw      & \gate{r}         & \qw \\
      \lstick{2}
      & \multigate{1}{M} \qwx & \qw & \ccteqg & & \multigate{1}{T} &
      \control \qw \qwx
      & \multigate{1}{S} & \control \qw \qwx  & \multigate{1}{R} & \qw \\
      \csl & \ghost{M} & \qw & \ccteqg & & \ghost{T} & \qw &
      \ghost{S} & \qw & \ghost{R} & \qw
    }
    \]
    Then either
    (I) an even number of $r,s,t$ are anti-diagonal,
    and the remainder diagonal, or (II) $S$ or $S\X^{(2)}$ commutes
    with $\Z^{(2)}$.
  \end{corollary}
  \begin{proof}
    In order to apply Lemma \ref{lem:sp2},
    We move $R$ and $T$ to the other side.
    \[
    \qc{ & \qw & \controlu \qw & \qw & \qw & \ccteq{2} & & \gate{t} &
      \control \qw
      & \gate{s}  & \qw       & \control \qw      & \gate{r}         & \qw \\
      \lstick{m} & \multigate{1}{T^\dagger} & \multigate{1}{M} \qwx &
      \multigate{1}{R^\dagger} & \qw & \ccteqg & & \qw & \control \qw
      \qwx
      & \qw & \multigate{1}{S} & \control \qw \qwx          & \qw & \qw \\
      \csl & \ghost{T^\dagger} & \ghost{M} & \ghost{R^\dagger} & \qw &
      \ccteqg & \csl & \qw & \qw & \qw & \ghost{S} & \qw & \qw & \qw }
    \]
    The cases here will correspond to the cases of Lemma \ref{lem:sp2}.
    Case II is preserved verbatim.  For Case I,
    the ``$a_i b_j$'' which must be diagonal are $rst,
    rs\Z t, r\Z st, r \Z s \Z t$.
    Since $(rst)^\dagger rs \Z t = t
    \Z t^\dagger$ is diagonal, we deduce that either $t$ or $t\X$ is diagonal.
    Likewise, $r \Z st (rst)^\dagger = r
    \Z r^\dagger$ is diagonal, so either $r$ or $r\X$ is diagonal.
    Finally, $rst$ is diagonal, so from what we know about $r,t$,
    either $s$ or $s \X$ is diagonal, and the number of $r,s,t$
    which are not diagonal is even.
  \end{proof}

  The following reformulation will be useful later.

  \begin{corollary} \label{cor:twoc}
    Suppose $Q$ commutes with $\Z^{(\ell)}$ and let
    $\mathcal{C}$ be a $\CZ^{(\ell)}$-circuit computing
    $Q$ in which exactly two \CZ s are incident on $\ell$,
    say $\CZ^{(\ell,m)}$ and $\CZ^{(\ell,n)}$.  Then
    all non-diagonal one-qubit gates may be eliminated
    from qubit $\ell$ at the cost of possibly
    (i) replacing $\CZ^{(\ell,n)}$ with $\CZ^{(\ell,m)}$
    and (ii) adding one-qubit gates on qubits $m, n$.
  \end{corollary}
  \begin{proof}
    By hypothesis, $\mathcal{C}$ takes the form
    \[Q = [r \otimes R]\CZ^{(\ell,m)}[s \otimes  S]
    \CZ^{(\ell,n)}[t \otimes T]\]
    where $r,s,t$ are subcircuits of one-qubit operators acting on
    $\ell$, and
    $R,S,T$ are subcircuits containing no gates
    acting on $\ell$.  We immediately replace $r,s,t$ by the one-qubit
    operators they compute.  Moreover, if
    $m \ne n$, then replace $S$
    and $T$ by $S \cdot {\tt SWAP}^{(m,n)}$ and
    ${\tt SWAP}^{(m,n)} \cdot T$, where ${\tt SWAP}$ is
    the gate which exchanges qubits.  The swaps will be
    restored and canceled at the end of the proof.
    We are in the situation of Lemma \ref{lem:sp2}.

    \vspace{1mm} \noindent {\bf Case I.}
    We are done, with the exception that the $r,s,t$ may be
    anti-diagonal rather than diagonal.  In this case,
    Equation \ref{circ:demorgan} allows the extraneous \X s
    to be pushed through and cancelled at the cost of
    introducing $\Z$ gates on qubit $m$.
    The diagonal gates remaining on
    qubit $\ell$ may be commuted through the
    $\CZ$s and conglomerated into one.
    Finally, the possible
    swap introduced between the $S, T$ terms
    may be cancelled.

    \vspace{1mm} \noindent {\bf Case II.}
    Using Equation \ref{circ:demorgan} and replacing $s$
    by $s\Z$ if necessary, we commute
    $S$ past one of the \CZ s.
    We now have:
    \[Q = [r \otimes R]\CZ^{(\ell,m)} s^{(\ell)}
    \CZ^{(\ell,m)}[t \otimes ST]\]
    Rearranging the equation,
    \begin{equation} \label{eq:V}
    [I \otimes R^\dagger] Q  [I \otimes T^{\dagger} S^\dagger] =
      r^{(\ell)} \CZ^{(\ell,m)} s^{(\ell)}
    \CZ^{(\ell,m)} t^{(\ell)}
    \end{equation}
    Let $V$ be the value of either side of the equation above.
    Then from the LHS we see that $V$ commutes with $Z^{(\ell)}$,
    and from the RHS we see that $V$ is a two-qubit operator
    commuting with $Z^{(m)}$.  Thus $V$ is a two-qubit diagonal,
    and admits the following decomposition.
    \[
    \qc{
      \lstick{\ell} & \multigate{1}{V} & \qw & \ccteq{1} & & \gate{R_z(\alpha)}
      & \qw & \control \qw & \qw & \qw & \qw & \control \qw & \qw & \qw \\
      \lstick{m} & \ghost{V} & \qw & \ccteqg & & \gate{R_z(\beta)} & \gate{H}
      & \control \qw \qwx & \gate{H} &
      \gate{R_z(\gamma)} & \gate{H} & \control \qw \qwx & \gate{H} & \qw }
    \]
    Substituting this decomposition for the RHS of Equation \ref{eq:V} and
    restoring the $R, S, T$ gates completes the proof.
  \end{proof}

\section{The \CNOT-cost of the \TOFFOLI\ gate}
\label{sec:toffoli}

  So far we have reduced \CNOT-counting for the \TOFFOLI\ gate
  to \CZ-counting for the \CCZ\ gate, with the latter two
  being diagonal and symmetric.  Having derived the inequality
  $3 |\CCZ|_{\CZ;\ell} / 2 \le  |\CCZ|_{\CZ}$, we seek to determine
  the qubit-local costs $|\CCZ|_{\CZ;\ell}$.

  The idea is to find an equivalence relation $\sim_{\ell}$ such that
  (i) $U \sim_{\ell} V \implies |U|_{\CZ;\ell} = |V|_{\CZ;\ell}$ and
  (ii) the equivalence classes of $\sim_\ell$ are easy to characterize.

  \begin{definition}
  For $P,Q$ commuting with $\Z^{(\ell)}$, we
  write $P \sim_{\ell} Q$ if there exist $a,b,A,B$ satisfying the
  following equation.
  \begin{equation} \label{eq:psimq}
    \qc{ \lstick{\ell} & \gate{b} & \controlu \qw & \gate{a} & \qw
      & \ccteq{1} & & \controlu \qw & \qw \\
      \csl & \gate{B} & \gate{P} \qwx & \gate{A} & \qw & \ccteqg & &
      \gate{Q} \qwx & \qw }
  \end{equation}
  \end{definition}

  The fact that $|\cdot|_{\CZ;\ell}$ is constant on equivalence
  classes is obvious; the ability to characterize the equivalence
  classes comes from a comparison between Equation \ref{eq:psimq}
  and the demultiplexing decomposition of Equation \ref{eq:demux}.
  We construct invariants of the equivalence classes in Theorem
  \ref{thm:loceq}.  The
  reductions of Section \ref{sec:avsc} provide circuit forms
  on which the invariants are easy to compute; as a consequence,
  we arrive at a complete characterization of $U$ such that
  $|U|_{\CZ;\ell} = 0, 1, 2$
  in Theorem \ref{thm:ccount}.
  The \CCZ\ gate falls into none of these classes,
  and thus  $|\CCZ|_{\CZ;\ell} \ge 3$, and hence $|\CCZ|_{\CZ} \ge 5$.
  Unfortunately, qubit-local
  \CZ -counting can take us no further:  one can show by construction
  that in fact $|\CCZ|_{\CZ;\ell} = 3$.

  We now consider a hypothetical
  five-\CZ\ circuit for the \CCZ\ and seek a contradiction,
  using a divide-and-conquer strategy.  There
  are many possible arrangements of the \CZ s, and we do not
  deal with them case by case.  Nonetheless,
  we fix one here for clarity.

  \begin{equation}
    \label{eq:lotsofletters}
    \qc{
      \lstick{1} & \control \qw & \qw & \ccteq{2}& &
      \gate{f} & \control \qw & \qw & \gate{e} &
      \qw & \control \qw & \qw & \gate{d} &
      \qw & \control \qw & \gate{c} &\qw \\
      \lstick{2} & \control \qw \qwx & \qw & \ccteqg& &
      \gate{k} & \control \qw \qwx & \gate{j} &  \control \qw &
      \qw & \qw \qwx & \gate{i} &  \control \qw &
      \gate{h} & \control \qw \qwx & \gate{g} & \qw \\
      \lstick{3} & \control \qw \qwx & \qw & \ccteqg& &
      \gate{o} & \qw & \qw &  \control \qw \qwx&
      \gate{n} & \control \qw \qwx  & \gate{m} & \control \qw \qwx &
      \qw & \qw & \gate{l} & \qw  \\
    }
  \end{equation}
  We define $a,b,P,Q$ as follows.

  \[
  \qc{
    & \gate{a}  & \qw & \ccteq{1} &  & \gate{d} & \control \qw
    & \gate{c} & \qw \\
    & \controlu \qwx \qw & \qw & \ccteqg  & & \qw & \control \qw \qwx
    & \qw & \qw }
  \]
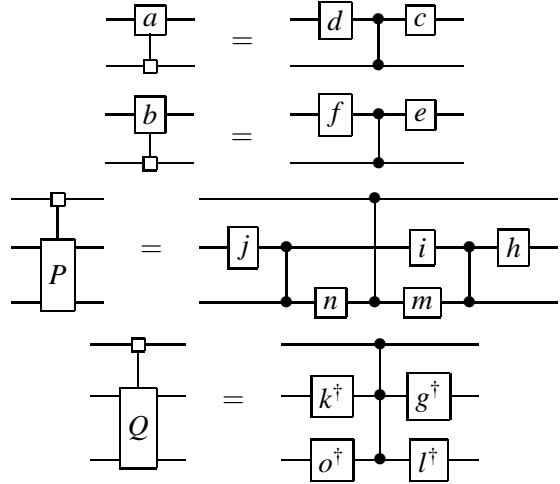
  \[
  \qc{
    & \gate{b} \qw & \qw & \ccteq{1} & & \gate{f} & \control \qw
    & \gate{e} & \qw \\
    & \controlu \qwx \qw & \qw & \ccteqg & & \qw & \control \qw \qwx
    & \qw & \qw
  }
  \]
  \[
  \qc{
    & \controlu \qw & \qw & \ccteq{2} & &
    \qw & \qw & \qw & \control \qw & \qw & \qw & \qw & \qw \\
    & \multigate{1}{P} \qwx & \qw & \ccteqg & &
    \gate{j} & \control \qw & \qw & \qw \qwx & \gate{i} & \control \qw
    & \gate{h} & \qw \\
    & \ghost{P} & \qw & \ccteqg & &
    \qw & \control \qw \qwx & \gate{n} & \control \qw \qwx
    & \gate{m} & \control \qw \qwx & \qw & \qw
  }
  \]
  \[
  \qc{
    & \controlu \qw & \qw & \ccteq{2} & & \qw & \control \qw & \qw & \qw \\
    & \multigate{1}{Q} \qwx & \qw & \ccteqg & & \gate{k^\dagger} &
    \control \qw \qwx
    & \gate{g^\dagger} & \qw \\
    & \ghost{Q} & \qw & \ccteqg & & \gate{o^\dagger} & \control \qw \qwx
    & \gate{l^\dagger} & \qw }
  \]
  Our circuit decomposition now takes the following form.

  \begin{equation}
    \label{eq:apb}
    \qc{
      \lstick{1} & \gate{b}    & \controlu \qw & \gate{a} & \qw
      & \ccteq{2} & & \controlu \qw & \qw \\
      \lstick{2} &  \controlu \qw \qwx & \multigate{1}{P} \qwx &
      \controlu \qw \qwx & \qw
      & \ccteqg & & \multigate{1}{Q} \qwx & \qw \\
      \csl & \qw & \ghost{P} & \qw & \qw & \ccteqg & & \ghost{Q} & \qw}
  \end{equation}

  Up to some two-qubit diagonal fudge factors, this equation says
  that the cosine-sine decomposition
  of $b^\dagger \otimes I$ is $Q^\dagger [a \otimes I] P$.
  In Section \ref{sec:avsc}, we translate the well-known
  canonicity of this Cartan decomposition into constraints on
  the components $a$, $b$, $P$ and $Q$.
  The formulae of Theorem
  \ref{thm:ccount} further strengthen these constraints in the
  $|\cdot|_{\CZ;\ell} = 3$ case.  Specifically, we show
  in Theorem \ref{thm:threec}
  that if $|U|_{\CZ;\ell} = 3$ and $\mathcal{C}$ computes $U$ using
  the minimum required three \CZ\ gates incident on $\ell$,
  then all one-qubit gates on $\ell$ are diagonal or anti-diagonal.  The
  anti-diagonal gates can be made diagonal at
  the cost of introducing \Z\ gates elsewhere in the circuit.

  This is the last result needed to determine the \CZ -cost of the
  \CCZ.  From $|\CCZ|_{\CZ;\ell} \ge 3$, we see that in any five-\CZ\
  circuit for the \CCZ, two of the qubits,
  $m,n$ touch exactly three \CZ\ gates and the remaining one touches four.
  By Theorem \ref{thm:threec},
  we can assume all one-qubit operators on $m,n$ are diagonal.
  Proposition \ref{prop:noninvariance} would then require
  $\det_{m,n} \CCZ = \CZ^{(m,n)}$ to be separable, which it is not.

  \begin{theorem} \label{thm:T6}
    $|\CCZ|_\CZ = 6$.
  \end{theorem}

  We show in
  Section \ref{sec:ancilla} that the use of ancillae
  can not lower the \CZ-cost of the \CCZ.

  \subsection{\CZ\ counting via the demultiplexing decomposition}
  \label{sec:invariants}

  We now turn to the study of qubit-local \CZ-cost.  To apply
  $P \sim_{\ell} Q \implies |P|_{\CZ;\ell} = |Q|_{\CZ;\ell}$, we
  first seek to determine when $P \sim_{\ell} Q$.
  This will be done under the assumption that $P$ and
  $Q$ both commute with $\Z^{(\ell)}$.

  \begin{definition}
    Let $U$ commute with $\Z^{(\ell)}$. Then the
    $\ell$-mux-spectrum $\inv^{(\ell)}(U)$ is the
    multi-set of eigenvalues, taken with
    multiplicity, of $U_1^\dagger U_0$.  Two multi-sets
    $S, T$ are said to be congruent, $S \cong T$, if there
    exists a nonzero scalar $\lambda$ such that either $\lambda
    S = T$ or $\lambda S =
    T^\dagger$. 
  \end{definition}

  We note that before taking the $\ell$-mux-spectrum of $U$,
  it is necessary to fix the number of qubits on which
  $U$ acts :
  $\inv^{(\ell)}(U \otimes I)$ contains $\dim I$ copies of
  $\inv^{(\ell)}(U)$.

  \begin{theorem} \label{thm:loceq}
    Suppose $P, Q$ commute with $\Z^{(\ell)}$.  Then $P \sim_\ell Q
    \iff \inv^{(\ell)}(P) \cong \inv^{(\ell)}(Q)$.
  \end{theorem}
  \begin{proof}
    $(\Rightarrow)$.  As $P \sim_{\ell} Q$, there are gates $a, b, A, B$ such
    that
    \[
    \qc{ \lstick{\ell} & \gate{b} & \controlu \qw & \gate{a} & \qw
      & \ccteq{1} & & \controlu \qw & \qw \\
      \csl & \gate{B} & \gate{P} \qwx & \gate{A} & \qw & \ccteqg & &
      \gate{Q} \qwx & \qw }
    \]
    By Corollary \ref{cor:sp1:CZ}, we
    may assume that either $a, b$ or $a \X, b\X$
    are diagonal.  In the first case,
    $Q_0 = a_0 b_0 A P_0 B$ and $Q_1 = a_1
    b_1 A P_1 B$. Thus $Q_1^\dagger Q_0 = (a_1 b_1)^\dagger a_0 b_0
    B^\dagger P_1^\dagger P_0 B$, which has the same eigenvalues as
    $(a_1 b_1)^\dagger a_0 b_0 P_1^\dagger P_0$. Thus
    $\inv^{(\ell)}(P) \cong \inv^{(\ell)}(Q)$.

    Otherwise, $a' = a \X$ and $b' = \X b$ are diagonal.
    Now $Q_0^\dagger Q_1 = (a'_1 b'_1)^\dagger a'_0 b'_0 B^\dagger
    P_0^\dagger P_1 B$, which has the same eigenvalues as $(a'_1
    b'_1)^\dagger a'_0 b'_0 P_0^\dagger P_1$, whose eigenvalues in turn are
    the complex conjugates of those of $a'_1 b'_1 (a'_0 b'_0)^\dagger
    P_1^\dagger P_0$; again $\inv^{(\ell)}(P) \cong \inv^{(\ell)}(Q)$.

    \vspace{2mm} $(\Leftarrow)$.  By supposition, the
    $\inv^{(\ell)}(P) \cong \inv^{(\ell)}(Q)$
    We note $\inv^{(\ell)}( \X^{(\ell)}P
    \X^{(\ell)} ) = \inv(P)^\dagger$ and $\inv( (R_z^{(\ell)}(\lambda)
    P ) = e^{2i\lambda} \inv(P)$.  Therefore we can readily find an
    operator $P' \sim_{\ell} P$ such that the
    $\ell$-mux-spectrum of $P$ is identical,
    rather than merely congruent, to that
    of $Q$.  It remains to show that $P' \sim_{\ell} Q$.

    By the demultiplexing decomposition (Equation \ref{eq:demux})
    there exist unitary operators $M_P, N_P$
    and a real diagonal matrix $\delta_P$, all of which
    operate on the qubits other than $\ell$, such that
    $P' = [I \otimes M_P] R_z^{(\ell)}(\delta_P) [I \otimes N_P]$.
    Likewise we
    decompose $Q = [I \otimes M_Q] R_z^{(\ell)}(\delta_Q) [I \otimes N_Q]$.
    If we let $\Delta_P = \exp(i \delta_P)$ and $\Delta_Q = \exp(i \delta_Q)$,
    then the $\ell$-mux-spectra of $P'$ and $Q$ are respectively the
    entries of $\Delta_P^2$ and $\Delta_Q^2$.
    Since $\inv^{(\ell)}(P)=\inv^{(\ell)}(Q)$, there
    must exist a permutation matrix $\pi$ acting on the qubits other than $\ell$
    such that $\pi \Delta_P^2 \pi^\dagger = \Delta_Q^2$.
    Rearranging, we have $\Delta_Q^\dagger \pi \Delta_P = \Delta_Q \pi
    \Delta_P^\dagger$.  Writing $K$ for this term,
    $[I \otimes M_Q K
    M_P^\dagger] P' [I \otimes N_P^\dagger \pi^\dagger N_Q] = Q$.
    Thus $P' \sim_{\ell} Q$.
  \end{proof}

  We now apply Theorem \ref{thm:loceq} to prove the following result
  relating $\inv^{(\ell)}(P)$ and $|P|_{\CZ;\ell}$.  We emphasize that
  the number of qubits on which $P$ acts is an unspecified
  parameter in both of these functions.

  \begin{theorem} \label{thm:ccount}
    Let $P$ commute with $\Z^{(\ell)}$.
    \begin{itemize}
    \item $|P|_{\CZ;\ell} = 0$ iff  $\inv^{(\ell)}(P) \cong \{1,1,\ldots\}$.
    \item $|P|_{\CZ;\ell} = 1$ iff  $\inv^{(\ell)}(P) \cong \{1,-1,1,-1,\ldots\}$
    \item $|P|_{\CZ;\ell} \le 2$ iff  $\inv^{(\ell)}(P)$ is congruent to some
      multi-set $S$ of unit norm complex numbers which come in
      conjugate pairs.
    \end{itemize}
  \end{theorem}
  \begin{proof}
    The first and second statements follow immediately from Theorem
    \ref{thm:loceq} and the calculations $\inv^{(\ell)}(I) = \{1,1,\ldots\}$
    and $\inv^{(\ell)}(\CZ^{(\ell,m)}) = \{1,-1,1,-1,\ldots\}$.
    To perform the relevant calculation for the third statement,
    we will use Corollary \ref{cor:twoc}.

    Let $\ell$ be the most significant qubit.
    For $\delta$ a diagonal real operator acting on all
    qubits but $\ell$, define $\Phi(\delta)$ by
    \[
    \qc{
      \lstick{\ell} & \controlu \qw & \qw & \ccteq{2} & & \qw & \control \qw
      & \qw     & \control \qw      & \qw         & \qw \\
      & \multigate{1}{\Phi(\delta)} \qwx & \qw & \ccteqg & & \qw &
      \control \qw \qwx & \gate{R_y(\delta)} & \control \qw \qwx
      & \qw & \qw \\
      \csl & \ghost{\Phi(\delta)} & \qw & \ccteqg & \csl & \qw & \qw &
      \controlu \qw \qwx & \qw & \qw & \qw }
    \]
    By construction, $|\Phi(\delta)|_{\CZ;\ell} \le 2$.  We compute
    $\inv^{(\ell)}(\Phi(\delta)) = \{e^{2i \delta_0}, e^{-2i\delta_0}, e^{2i\delta_1},
    e^{-2i\delta_1}, \ldots, \}$.

    $(\Leftarrow)$ Write the entries of $S$ as $
    e^{i \phi}\cdot\{e^{i \theta_0},
    e^{-i \theta_0}, e^{i \theta_1}, e^{-i \theta_1}, \ldots \}$, and
    let $\theta$ be the real diagonal operator acting on
    all qubits but $\ell$
    whose diagonal entries
    are $\theta_0, \theta_1, \ldots$.  By construction,
    $\inv^{(\ell)}(\Phi(\theta/2)) = S$, and $S \cong \inv^{(\ell)}(Q)$
    by hypothesis.  By Theorem \ref{thm:loceq},
    $\Phi(\theta/2) \sim_{\ell} Q$ are $\ell$-equivalent.  It
    follows that $|Q|_{\CZ;\ell} = |\Phi(\theta/2)|_{\CZ;\ell} \le 2$.

    $(\Rightarrow)$ By hypothesis $|Q|_{\CZ;\ell} \le 2$.  If in
    fact $|Q|_{\CZ;\ell} = 0, 1$, note by the first two statements
    of the Theorem, which have been proven,
    the $\ell$-mux-spectrum of $Q$ has the desired property.
    Thus we assume $|Q|_{\CZ;\ell} = 2$.  Let
    $\mathcal{C}$ be a circuit in which this minimal \CZ\ count is
    achieved.  By Corollary \ref{cor:twoc}, we can find an equivalent
    circuit $\mathcal{C'}$ of the following form.
    \[
    \qc{ & \controlu \qw & \qw & \ccteq{3} & & \gate{R_z(\theta)} &
      \control \qw
      & \qw         & \control \qw      & \qw         & \qw \\
      & \multigate{2}{Q} \qwx & \qw & \ccteqg & & \multigate{2}{A} &
      \control \qw \qwx
      & \multigate{2}{B} & \qw \qwx          & \multigate{2}{C} & \qw \\
      & \ghost{Q} & \qw & \ccteqg & & \ghost{A} & \qw
      & \ghost{B}        & \control \qw \qwx & \ghost{C}        & \qw \\
      \csl & \ghost{Q} & \qw & \ccteqg & \csl & \ghost{A} & \qw &
      \ghost{B} & \qw & \ghost{C} & \qw }
    \]
    We have drawn the \CZ s with different lower contacts, but of
    course they might be the same.  Actually, we prefer the latter
    case, and ensure it by incorporating swaps into $B, C$ if
    necessary.  We take a cosine-sine
    decomposition (see Equation \ref{eq:csd}) of $B$
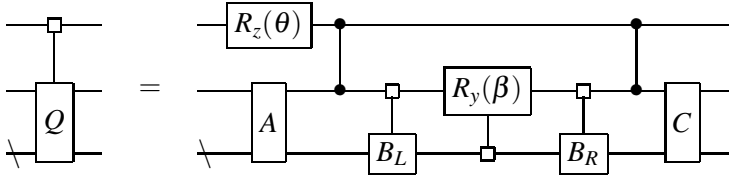
    \[
    \qc{ & \controlu \qw & \qw & \ccteq{2} & & \gate{R_z(\theta)} &
      \control \qw
      & \qw & \qw & \qw        & \control \qw      & \qw         & \qw \\
      & \multigate{1}{Q} \qwx & \qw & \ccteqg & & \multigate{1}{A} &
      \control \qw \qwx & \controlu \qw & \gate{R_y(\beta)} &
      \controlu \qw & \control \qw \qwx
      & \multigate{1}{C} & \qw \\
      \csl & \ghost{Q} & \qw & \ccteqg & \csl & \ghost{A} & \qw &
      \gate{B_L} \qwx & \controlu \qw \qwx & \gate{B_R} \qwx & \qw &
      \ghost{C} & \qw }
    \]
    Note that the $B_L$ and $B_R$ gates commute with the \CZ s.  Thus
    $Q \sim_{\ell} \Phi(\beta)$.  By Theorem
    \ref{thm:loceq}, the $\inv^{(\ell)}(Q) \cong \inv^{(\ell)}(\Phi(\beta))$.
    But we have already seen that
    $\inv^{(\ell)}(\Phi(\cdot))$ always consists of conjugate
    pairs of unit-norm complex numbers.
  \end{proof}

\subsection{Circuit constraints from the cosine-sine decomposition}
\label{sec:avsc}

  This section is devoted to the study of Equation \ref{eq:apb}.
  We take cosine-sine decompositions of $a,b$.
  Below, $A_l,A_r,B_l,B_r$ are two-qubit diagonal operators, and
  $\alpha, \beta$ are $2\times 2$ real diagonal matrices of angular parameters.

  \begin{equation}\label{eq:csd:b}
    \qc{
      \lstick{1} & \gate{b} & \qw & \ccteq{1} & \multigate{1}{B_L}
      & \gate{R_y(-\beta)}
      & \multigate{1}{B_R} & \qw \\
      \lstick{2} & \controlu \qw \qwx & \qw & \ccteqg & \ghost{B_L}
      & \controlu \qw \qwx
      & \ghost{B_R} & \qw \\
    }
  \end{equation}

  \begin{equation} \label{eq:csd:a}
    \qc{
      \lstick{1} & \gate{a} & \qw & \ccteq{1} & \multigate{1}{A_L}
      & \gate{R_y(\alpha)}
      & \multigate{1}{A_R} & \qw \\
      \lstick{2} & \controlu \qw \qwx & \qw & \ccteqg
      & \ghost{A_L} & \controlu \qw \qwx
      & \ghost{A_R} & \qw \\
    }
  \end{equation}
  Define $\tilde{P} = A_L P B_R$ and $\tilde{Q} = A_R^\dagger Q B_L^\dagger$
  to obtain:

  \begin{equation}
    \qc{
      \lstick{1} & \gate{R_y(-\beta)}    & \controlu \qw
      & \gate{R_y(\alpha)} & \qw
      & \ccteq{2} & & \controlu \qw & \qw \\
      \lstick{2} & \controlu \qw \qwx & \multigate{1}{\tilde{P}} \qwx
      & \controlu \qw \qwx & \qw
      & \ccteqg & & \multigate{1}{\tilde{Q}} \qwx & \qw \\
      \csl & \qw & \ghost{\tilde{P}} & \qw & \qw & \ccteqg &
      & \ghost{\tilde{Q}} & \qw}
  \end{equation}

  We recall the standard argument used to measure the uniqueness of the
  KAK decomposition \cite{Knapp:98}.
  Throughout this discussion, we will write simply $R_y(\alpha)$
  for $R_y^{(1)}(\alpha^{(2)})$, and similarly for $R_y(\beta)$.
  Rearrange the equation to obtain
  $\tilde{Q}^\dagger R_y(\alpha) \tilde{P} = R_y(\beta)$.
  Transforming the equation by $k \mapsto \Z^{(1)} k^\dagger \Z^{(1)}$,
  we get $\tilde{P}^\dagger R_y(\alpha)
  \tilde{Q} = R_y(\beta)$.  Multiplying
  these equations yields
  $\tilde{P}^\dagger R_y(2\alpha) \tilde{P} =
  R_y(2\beta)$.
  Thus $R_y(2\alpha)$ and $R_y(2 \beta)$
  have the same eigenvalues.  One can check that
  in fact they are conjugate under
  an element of the 
  group $W$ generated by $\X^{(2)}$ and $\CZ^{(1,2)}$;
  note that these operators commute with $\Z^{(1)}$.  That is, there exists
  $w \in W$ such that $w R_y(2\alpha) w^\dagger = R_y(2\beta)$.  Now let
  $t = w R_y(\alpha) w^\dagger R_y(-\beta)$.  We have both
  $t = R_y(\xi)$ for some
  $2\times 2$ real diagonal matrix $\xi$ acting on qubit 2, and
  $t^2 = I$; it follows that
  $t \in \{\pm I, \pm \Z^{(2)}\}$.  Defining
  $\bar{P} = \tilde{P}\cdot [t w \otimes I]$ and
  $\bar{Q} = \tilde{Q}\cdot [w \otimes I]$ reduces
  our equation to the following.

  \begin{equation} \label{eq:alphacomm}
    \qc{
      \lstick{1} & \gate{R_y(-\alpha)}    & \controlu \qw
      & \gate{R_y(\alpha)} & \qw
      & \ccteq{2} & & \controlu \qw & \qw \\
      \lstick{2} & \controlu \qw \qwx & \multigate{1}{\bar{P}} \qwx
      & \controlu \qw \qwx & \qw
      & \ccteqg & & \multigate{1}{\bar{Q}} \qwx & \qw \\
      \csl & \qw & \ghost{\bar{P}} & \qw & \qw & \ccteqg &
      & \ghost{\bar{Q}} & \qw}
  \end{equation}

  By an argument similar to that
  given for $\tilde{P}$ and $\tilde{Q}$, the
  operators $\bar{P}$ and $\bar{Q}$ both commute with
  $R_y(2\alpha)$.  Conjugation by $R_y(\alpha)$ is
  an involution on the set of operators commuting with
  $R_y(2\alpha)$; Equation \ref{eq:alphacomm} says that
  $P$ and $Q$ are interchanged by this involution.
  In fact, this involution always has a simpler description:

  \begin{lemma}  \label{lem:conjredux}
    Equation \ref{eq:alphacomm} also holds for
    some $\tilde{\alpha}$ for which
    $\tilde{\alpha}_i$ is an integer or half-integer
    multiple of $\pi$.  Half-integers occur
    if and only if $2 \alpha_i$ is an odd integer multiple of $\pi$.
  \end{lemma}
  \begin{proof}
    Decompose $2\alpha_i = \phi_i + \psi_i \pmod{2 \pi}$
    where $\phi_i \in (-\pi, \pi)$, where $\psi_i = 0$ unless
    $\phi_i = 0$, and $\psi_i \in \{0, \pi\}$ in any event.  Then
    any operator which commutes with $R_y(2 \alpha)$ also commutes with
    $R_y(\phi / 2)$.  Thus, on operators commuting with $R_y(2 \alpha)$,
    conjugation by $R_y(\alpha)$
    is the same as conjugation by $R_y(\alpha - \phi/2)
    = R_y(\alpha - \phi/2 - \psi/2) R_y(\psi/2)$.
    But $2(\alpha - \phi/2 - \psi/2) = 0 \pmod{2 \pi}$.
  \end{proof}

  We also record the constraints imposed on possible $\bar{P},
  \bar{Q}$ by the value of $\theta = 2 \alpha$.

  \begin{lemma}
    \label{lem:comm}
    Fix distinct qubits $\ell, m$.
    Let $U$ be a unitary operator commuting
    with $\Z^{(\ell)}$,
    and let $\theta$
    be a two-by-two real diagonal matrix of angular
    parameters which is understood
    to operate on $m$.  Then
    $U$ commutes with $R_y^{(\ell)}(\theta)$ if and
    only if one of the following holds:
    \begin{enumerate}
    \item $\cos(\theta)$ is scalar, and either
      \begin{enumerate}
      \item $\sin(\theta) = 0$.
      \item $\sin(\theta)$ is a nonzero scalar and $U_0 = U_1$.
      \item $\Z \sin(\theta)$ is a nonzero scalar and
        $U_0 = \Z^{(m)} U_1 \Z^{(m)}$.
      \end{enumerate}
    \item $\cos(\theta)$ is not scalar, $U$ commutes with $\Z^{(m)}$,
      and either
      \begin{enumerate}
      \item $\sin(\theta_0)= 0$ and $\sin(\theta_1) = 0$.
      \item $\sin(\theta_0)= 0$ and $\sin(\theta_1) \ne 0$
        and $U_{01} = U_{11}$.
      \item $\sin(\theta_0) \ne 0$ and $\sin(\theta_1) = 0$
        and $U_{00} = U_{10}$.
      \item $\sin(\theta_0) \ne 0$ and $\sin(\theta_1) \ne 0$
        and $U_0 = U_1$.
      \end{enumerate}
    \end{enumerate}
  \end{lemma}
  \begin{proof}
    The $(\Leftarrow)$ direction is trivial.
    For $(\Rightarrow)$, suppose $[R_y^{(\ell)}(\theta^{(m)}), U] = 0$
    and expand using the expression
    $R_y^{(\ell)}(\theta^{(m)}) = \exp (i \Y^{(\ell)} \theta^{(m)}) =
    \cos(\theta)^{(m)}
    + i Y^{(\ell)} \sin(\theta)^{(m)}$ in order to observe
    that $U_0$ and $U_1$ both commute with
    $\cos( \theta)^{(m)}$, and
    $U_0 \sin(\theta)^{(m)} =  \sin(\theta)^{(m)} U_1$.
    Now repeatedly
    apply the fact that two-by-two matrices
    which commute with a two-by-two diagonal matrix with distinct
    entries are themselves diagonal.
  \end{proof}

  Finally, we translate these results back to the original
  operators $P,Q$.

  \begin{lemma} \label{lem:sp}
    In the situation of Equation \ref{eq:apb}, at least
    one of the following must hold.
    \begin{enumerate}
      \item Either $a,b$ are diagonal or $a\X^{(1)},b\X^{(1)}$ are diagonal.
      \item There exists a two-qubit operator $U$
        and two-qubit diagonals $D, D'$ such that
        \[
        \qc{
          & \controlu \qw & \qw & \ccteq{2} &
          & \multigate{1}{D'} & \qw & \multigate{1}{D} & \qw \\
          & \multigate{1}{P} \qwx \qw & \qw & \ccteqg &
          & \ghost{D'} & \multigate{1}{U} & \ghost{D} & \qw \\
          & \ghost{P} & \qw & \ccteqg & & \qw & \ghost{U} & \qw & \qw
        }
        \]

        Similarly, there exists a two-qubit operator $V$
        and two-qubit diagonals $C, C'$ such that
        \[
        \qc{
          & \controlu \qw & \qw & \ccteq{2} &
          & \multigate{1}{C'} & \qw & \multigate{1}{C} & \qw \\
          & \multigate{1}{Q} \qwx \qw & \qw & \ccteqg &
          & \ghost{C'} & \multigate{1}{V} & \ghost{C} & \qw \\
          & \ghost{Q} & \qw & \ccteqg & & \qw & \ghost{V} & \qw & \qw
        }
        \]
      \item Either $P$ or $P\X^{(2)}$ commute with $\Z^{(2)}$.
        There exist replacements $a', b'$ for $a,b$ which
        are in the subgroup generated by
        two-qubit diagonal operators
        on qubits 1 and 2, $\C^{(2)}\X^{(1)}$, and $\X^{(1)}$, such that
        Equation \ref{eq:apb} continues to hold.
    \end{enumerate}
  \end{lemma}
  \begin{proof}
    This amounts to unwinding the above discussion in light of Lemma
    \ref{lem:comm}.
    Case I comes from Case 1.a of the Lemma;
    the $\X$ appears because of the $2$ in $\theta = 2\alpha$.
    Case II comes from Cases 1.b and 1.c.
    The first claim in Case III
    is just Case 2 of the Lemma; the possible
    $\X$ here comes from the $w$ factor
    in $\bar{P} = \tilde{P}t w$ from the discussion above.
    The second claim follows from Lemma \ref{lem:conjredux}.
  \end{proof}

  While we cannot completely characterize operators with $|\cdot|_{\CZ;\ell}=3$,
  we can characterize $\CZ^{(\ell)}$-minimal
  circuits which compute them.

  \begin{theorem} \label{thm:threec}
    Fix a qubit $\ell$, and suppose
    $M$ commutes with $\Z^{(\ell)}$.  Suppose $|M|_{\CZ;j} = 3$, and
    let $\mathcal{C}$ be a $\CZ^{(j)}$-circuit exhibiting this bound.
    Then all one-qubit gates of $\mathcal{C}$ on $\ell$ are diagonal
    or anti-diagonal.
  \end{theorem}
  \begin{proof}
    Consider $M,\mathcal{C}$ satisfying the hypothesis.
    Without loss of
    generality, $\ell = 1$ and $\mathcal{C}$ takes the form
    \[
    \qc{ \lstick{1} & \gate{h} & \control \qw & \gate{g} & \control
      \qw
      & \gate{f} & \control \qw & \gate{e} & \qw\\
      \lstick{2} & \multigate{1}{H} & \control \qw \qwx &
      \multigate{1}{G} & \control \qw \qwx
      & \multigate{1}{F} & \control \qw \qwx & \multigate{1}{E} & \qw  \\
      \csl & \ghost{H} & \qw & \ghost{G} & \qw & \ghost{F} & \qw
      & \ghost{E} & \qw &  \\
    }
    \]
    The \CZ s may have originally had different terminals, but we can
    incorporate swaps into $E, F, G, H$ to suppress this behavior.
    This affects neither the hypothesis nor the conclusion.

    (*) Define $P$ by
    \[
    \qc{
      \lstick{1} & \controlu \qw & \qw & \ccteq{2} &
      & \qw & \control \qw & \qw & \qw \\
      \lstick{2} & \multigate{1}{P} \qwx & \qw & \ccteqg &
      & \multigate{1}{G} & \control \qw \qwx
      & \multigate{1}{F} & \qw \\
      \csl & \ghost{P} & \qw & \ccteqg &
      & \ghost{G} & \qw & \ghost{F} & \qw }
    \]
    If $P \X^{(2)}$ commutes with $\Z^{(2)}$, then
    return to (*) and replace
    $G$ by $G \X^{(2)}$, $H$ by $\X^{(2)} H$, and
    $h$ by $Z^{(1)} h$.  This does not
    affect the conclusion, and by Equation \ref{circ:demorgan},
    the resulting circuit still computes $M$.
    We have ensured that if one of $P, P \X^{(2)}$
    commutes with $\Z^{(2)}$, then it is $P$.

    Define $a, b, Q$ by
    \[
    \qc{
      \lstick{1} & \gate{a}  & \qw & \ccteq{1} &  & \gate{f} & \control \qw
      & \gate{e} & \qw \\
      \lstick{2}
      & \controlu \qwx \qw & \qw & \ccteqg  & & \qw & \control \qw \qwx
      & \qw & \qw }
    \]
    \[
    \qc{
      \lstick{1} & \gate{b} \qw & \qw & \ccteq{1} & & \gate{h} & \control \qw
      & \gate{g} & \qw \\
      \lstick{2}
      & \controlu \qwx \qw & \qw & \ccteqg & & \qw & \control \qw \qwx
      & \qw & \qw }
    \]
    \[
    \qc{
      \lstick{1} & \controlu \qw & \qw & \ccteq{2} &
      & \qw & \controlu \qw & \qw & \qw \\
      \lstick{2} & \multigate{1}{Q} \qwx & \qw & \ccteqg &
      & \multigate{1}{H^\dagger} &
      \multigate{1}{M} \qwx & \multigate{1}{E^\dagger} & \qw \\
      \csl & \ghost{Q} & \qw & \ccteqg &
      & \ghost{H^\dagger} & \ghost{M} & \ghost{E^\dagger} & \qw
    }
    \]
    Note $|Q|_{\CZ;1} = |M|_{\CZ;1}$.  We also have
    $Q = [a \otimes I]P[b \otimes I]$, hence
    are in the situation of Equation \ref{eq:apb}.
    Lemma \ref{lem:sp} allows us to reduce to the following cases.

    \noindent {\bf Case I.}
    $a, b$ are diagonal, or $a\X^{(1)}, b\X^{(1)}$ are diagonal.
    In either case, Corollary \ref{cor:sp1:CZ}
    applied to the circuits defining $a,b$
    shows that $e,f,g,h$ are each diagonal or anti-diagonal.

    \noindent {\bf Case II.}
    $Q$ takes the form
    \[
    \qc{
      \lstick{1}& \controlu \qw & \qw & \ccteq{2} &
      & \multigate{1}{C'} & \qw & \multigate{1}{C} & \qw \\
      \lstick{2} & \multigate{1}{Q} \qwx \qw & \qw & \ccteqg &
      & \ghost{C'} & \multigate{1}{V} & \ghost{C} & \qw \\
      \csl & \ghost{Q} & \qw & \ccteqg & & \qw & \ghost{V} & \qw & \qw
    }
    \]
    The cosine-sine decomposition (see Equation \ref{eq:csd})
    of $V$ along qubit $2$ determines
    unitary operators $R,S$ and a real diagonal operator $\delta$
    such that:
    \begin{equation}
      \qc{
        \lstick{2} & \multigate{1}{V} & \qw & \ccteq{1} & \controlu \qw
        & \gate{R_y(\delta)}
        & \controlu \qw & \qw \\
        \csl & \ghost{V} & \qw & \ccteqg & \gate{S} \qwx
        & \controlu \qw \qwx
        & \gate{R} \qwx & \qw \\
      }
    \end{equation}
    We substitute, commute the $S,T$ outwards past
    $C, C'$, and decompose the diagonals $C, C'$.
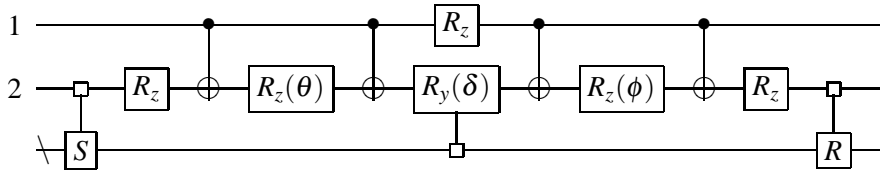
    \[
    \qc{
      \lstick{1} & \qw & \qw & \control \qw & \qw & \control \qw & \gate{R_z}&
      \control \qw & \qw
      & \control \qw & \qw & \qw & \qw \\
      \lstick{2} & \controlu \qw & \gate{R_z} & \targ \qwx
      & \gate{R_z(\theta)} &
      \targ \qwx & \gate{R_y(\delta)}
      & \targ \qwx & \gate{R_z(\phi)} & \targ \qwx & \gate{R_z}
      & \controlu \qw & \qw \\
      \csl & \gate{S} \qwx & \qw & \qw & \qw & \qw & \controlu \qw \qwx &
      \qw & \qw & \qw & \qw & \gate{R} \qwx & \qw }
    \]
    Evidently $\inv^{(1)}(Q)$ depends only on
    $\theta, \delta, \phi$.  We calculate that, up to
    a global scalar multiple, $\inv^{(1)}(Q)$ consists of the roots of
    the following quadratics in $T$:
    \[
    T^2 - 2 T (\cos(2\theta + 2 \phi) \cos (\delta_i)^2
    + \cos(2\theta - 2 \phi) \sin (\delta_i )^2 ) + 1
    \]
    The equations being real, each has complex conjugate roots.
    By Theorem \ref{thm:ccount}, $ |M|_{\CZ;1} = |Q|_{\CZ;1} = 2$,
    contrary to hypothesis.

    \noindent {\bf Case III.}
    We have already ensured that
    $P$, rather than $P \X^{(2)}$, commutes with
    $\Z^{(2)}$.  We replace $a, b$ by the $a', b'$ of
    Lemma \ref{lem:sp}.
    We demultiplex $P$ (see Equation \ref{eq:demux}) to obtain
    a decomposition of the following form, where $D$ is diagonal.
    \[
    \qc{
      \lstick{1} & \controlu \qw & \qw & \ccteq{2} & & \qw & \multigate{2}{D}
      \qw
      & \qw & \qw  \\
      \lstick{2} & \controlu \qw \qwx & \qw & \ccteqg & & \controlu \qw &
      \ghost{D}
      & \controlu \qw & \qw \\
      \csl & \gate{P} \qwx & \qw & \ccteqg & & \gate{S} \qwx &
      \ghost{D} & \gate{R} \qwx & \qw }
    \]

    The operators $S, R$ commute past $a', b'$ to the
    edges of the circuit, and thus do not affect
    the \CZ -cost of $Q$.  That is,
    $|Q|_{\CZ;\ell} = |[a' \otimes I]D[b'\otimes I]|_{\CZ;\ell}$.

    By construction, $|P|_{\CZ;\ell} = |D|_{\CZ;\ell} = 1$.
    If
    $D = \ketbra{0}{0}^{(\ell)} \otimes D_0 + \ketbra{1}{1}^{(\ell)}\otimes D_1$,
    Theorem \ref{thm:ccount} asserts
    the entries of $D_0^\dagger D_1$ are $e^{i \theta} \{1,-1,1,-1,\ldots\}$.
    Thus $D$ can be written as
    \[
    \qc{
      \lstick{1} & \multigate{2}{D} & \qw & \ccteq{2} & & \gate{R_z(-
        \theta/2)} &
      \control \qw      & \qw & \qw & \qw \\
      \lstick{2} & \ghost{D} & \qw & \ccteqg & & \multigate{1}{\pi} & \control
      \qw \qwx & \multigate{1}{\pi^\dagger} & \multigate{1}{D_0} & \qw
      \\
      \csl & \ghost{D} & \qw & \ccteqg & & \ghost{\pi} & \qw &
      \ghost{\pi^\dagger} & \ghost{D_0} & \qw }
    \]
    for some permutation $\pi$.
    We set $N := \X^{(1)} ([a' \otimes I] D [b' \otimes I])^{\dagger}
    \X^{(1)} [a' \otimes I] D [b' \otimes I]$, so that
    $\inv^{(1)}([a' \otimes I] D [b' \otimes I])$ is given by
    the entries of $\bra{0}^{(1)}N\ket{0}^{(1)}$.
    Evidently $D_0$ commutes past
    $a'$ and cancels with $D_0^\dagger$.
    Applying Equation \ref{circ:demorgan} to eliminate
    $\X$ gates, the following
    circuit computes $N$.
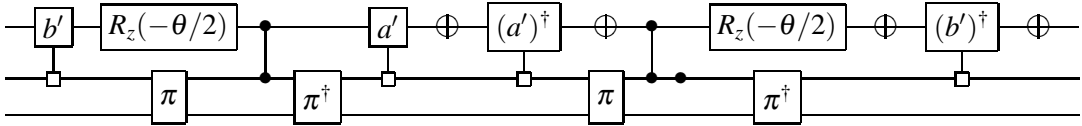
    \[
    \qc{
      & \gate{b'} & \gate{R_z(-\theta/2)} & \control \qw & \qw &
      \gate{a'} & \targ & \gate{(a')^\dagger} & \targ & \control \qw &
      \qw & \gate{R_z(-\theta/2)}
      & \targ & \gate{(b')^\dagger}  & \targ & \qw  \\
      & \controlu \qwx \qw & \multigate{1}{\pi} & \control \qw \qwx &
      \multigate{1}{\pi^\dagger} & \controlu \qw \qwx & \qw
      & \controlu \qwx \qw & \multigate{1}{\pi} & \control \qw \qwx
      & \control \qw &  \multigate{1}{\pi^\dagger} & \qw
      & \controlu \qw \qwx  & \qw & \qw \\
      & \qw & \ghost{\pi} & \qw & \ghost{\pi^\dagger} & \qw & \qw &
      \qw & \ghost{\pi} & \qw & \qw & \ghost{\pi^\dagger} & \qw & \qw
      & \qw & \qw }
    \]
    The condition on $a'$ implies that
    $(a')^\dagger\X^{(1)} a'
    \X^{(1)}$ is diagonal.  It follows that the subcircuit
    sandwiched between the two \CZ s computes a diagonal operator, and
    so the \CZ s cancel.  Then the $\pi$, $\pi^\dagger$ pair on the
    left cancel.  The $\pi^\dagger \Z^{(m)} \pi$ term on the
    right commutes past the $(b')^\dagger$.  What remains is a
    circuit of the form
    \[
    \qc{
      \lstick{1} & \multigate{1}{F} & \qw & \qw & \qw & \qw \\
      \lstick{2} & \ghost{F} & \multigate{1}{\pi} & \control \qw
      & \multigate{1}{\pi^\dagger} & \qw \\
      \csl & \qw & \ghost{\pi} & \qw & \ghost{\pi^\dagger} & \qw }
    \]
    By construction, $N$ commutes with both $\Z^{(1)}$ and $\Z^{(2)}$.
    It follows that $F$ is diagonal.  Then $f = \bra{0}^{(1)} F
    \ket{0}^{(1)}$ is some one-qubit diagonal acting on $m$.  We
    have $\bra{0}^{(1)} N \ket{0}^{(1)} = \pi^\dagger
    \Z^{(2)} \pi f^{(2)}$.  Denote by $f_0, f_1$ the entries of
    $f$.  Then the entries of $\bra{0}^{(1)} N \ket{0}^{(1)}$
    are $f_0, f_1, -f_0,
    -f_1$, and moreover $f_0$ will occur with the same multiplicity as
    $-f_1$; likewise $-f_0$ will occur with the same multiplicity as
    $f_1$.
    We see that $\sqrt{-f_0/f_1} \inv^{(1)}([a' \otimes I] D [b' \otimes I])$
    come in
    conjugate pairs.  By Theorem \ref{thm:ccount},
    $|[a' \otimes I] D [b' \otimes I]|_{CZ;1}\le 2$.  But now
    $|M|_{CZ;1} = |Q|_{CZ;1} = |[a' \otimes I] D [b' \otimes I]|_{CZ;1}$,
    contrary to hypothesis.
  \end{proof}

  \subsection{Corollaries}
  \label{sec:corollaries}

  The ${\PERES}$ gate implements a three-qubit transformation
  from classical reversible logic ${\PERES}^{(\ell;m;n)} = \C^{(\ell)} \X^{(m)} \cdot \C\C^{(\ell,m)}X^{(n)}$.
  As shown in \cite{Maslov:03}, it  can be a useful alternative to the
  \TOFFOLI\ gate in reversible circuits.

  \begin{corollary} \label{cor:peres}
    $|\PERES|_{\CZ} = 5$.
  \end{corollary}
  \begin{proof}
    As is clear from its definition, the $\PERES$ gate can
    be implemented by the circuit of Figure \ref{fig:6CNOTs}, save
    the rightmost $\CNOT$.  Thus,
    $|\PERES|_{\CZ} \le 5$.
    On the other hand, it also follows from
    the definition that any circuit for the \PERES\ can, with the addition
    of a single \CNOT , become a circuit for the \TOFFOLI.  Thus
    $|\PERES|_{\CZ} \ge |\TOFFOLI|_{\CZ} - 1 = 5$, and all inequalities
    are equalities.
  \end{proof}

  In a different direction,
  we consider below multiply-controlled $\Z$ gates:

  \begin{corollary} \label{cor:multicz}
    $|(n-1)-\mathrm{controlled}-\Z|_{\CZ} \ge 2n$ for any $n \ge 3$.
  \end{corollary}
  \begin{proof}
    We proceed by induction on $n$.  Suppose the Corollary is false;
    choose minimal falsifying $n$, and a falsifying circuit
    $\mathcal{C}$.  By Theorem \ref{thm:T6}, $n > 2$.  As before,
    at least three \CZ\ gates are incident to each qubit, and
    counting shows that at least one, say $\ell$ touches exactly three.
    As before, we can assume that all one-qubit
    operators which appear on $\ell$ are diagonal.
    Form the circuit
    $\mathcal{C}' = \bra{1}^{(\ell)} \mathcal{C} \ket{1}^{(\ell)}$
    by replacing every gate $g$ of $\mathcal{C}$ with
    $g' = \bra{1}^{(\ell)} g \ket{1}^{(\ell)}$.  This has no effect on
    gates which do not touch $\ell$; it turns one-qubit gates on $\ell$
    into scalars, and replaces $\CZ^{(\ell,s)}$ with $\Z^{(s)}$.
    At any rate, $\mathcal{C}'$ is a \CZ-circuit on $(n-1)$ qubits
    which computes the $(n-2)$-controlled-$\Z$.  We
    deduce by induction that it
    contains at least $2(n-1)$ \CZ\ gates.  Adding the (at least) three \CZ s
    incident to $\ell$, there are at least $2n + 1$ total \CZ s
    in $\mathcal{C}$.
  \end{proof}

  \section{Three-qubit diagonal operators}
  \label{sec:3diag}

  We give here a complete classification of three-qubit diagonal
  operators by their \CZ -cost.  Throughout this section, we
  assume no ancillae are available and
  label our qubits 1, 2, 3, from most significant to least
  significant.
  We abbreviate
  $\bra{i}^{(1)}\bra{j}^{(2)}\bra{k}^{(3)} D \ket{i}^{(1)}\ket{j}^{(2)}\ket{k}^{(3)}$
  by $D_{ijk}$.  We also write
  $\Delta(\eta)$ for the one-qubit gate given by
  $\ketbra{0}{0} + \ketbra{1}{1} \eta$.  Define
  \[\lambda_1(D) = \frac{D_{011}D_{000}}{D_{001}D_{010}},~~~~~~~
  \lambda_2(D) =  \frac{D_{101}D_{000}}{D_{100}D_{001}} ,~~~~~~~
  \lambda_3(D) = \frac{D_{110}D_{000}}{D_{100}D_{010}}  ,~~~~~~~
  \xi(D) = \frac{D_{111}D_{000}^2}{D_{100}D_{010}D_{001}}\]
  Then any three-qubit diagonal $D$
  admits the expansion
  \[D = D_{000} \cdot
  \Delta\left(\frac{D_{100}}{D_{000}}\right)^{(1)}
   \cdot
  \Delta\left(\frac{D_{010}}{D_{000}}\right)^{(2)}
  \cdot
  \Delta\left(\frac{D_{001}}{D_{000}}\right)^{(3)}
  \cdot
  \mbox{diag}(1,1,1,\lambda_1(D),1,\lambda_2(D), \lambda_3(D), \xi(D))
  \]
  The $\lambda_i(D)$ are multiplicative, $\lambda_i(D D')
  = \lambda_i(D) \lambda_i(D')$, and likewise for $\xi$.
  We denote by $S(D)$ the ordered quadruple
  $(\lambda_1(D),\lambda_2(D),\lambda_3(D), \xi(D))$.

  \begin{observation} \label{obs:diag:loceq}
  For $D, D'$ three-qubit diagonal operators,
  $S(D) = S(D')$ iff
  $S(D^\dagger D') = (1,1,1,1)$ iff $D^\dagger D'$ is a tensor
  product of one-qubit diagonal operators. It follows that
  $S(D) = S(D') \implies |D|_{\CZ;i} = |D'|_{\CZ;i}$.
  \end{observation}

  \begin{observation} \label{obs:inv:s}
    $\inv^{(i)}(D) = \{1,\lambda_j(D)^\dagger, \lambda_k(D)^\dagger,
    \xi(D)^\dagger \lambda_i(D)\}$ where $\{i,j,k\} = \{1,2,3\}$.
  \end{observation}

  \begin{lemma} \label{lem:3qd}
    A three-qubit diagonal $D$ can be implemented
    in a three-qubit \CZ -circuit with:
    \begin{itemize}
      \item 0 \CZ s on touching qubit 1 iff $S(D) = (\xi,1,1;\xi)$
      \item 1 \CZ\ touching qubit 1 iff $S(D) = (\xi,-1,-1,;\xi),
        (-\xi,1,-1;\xi), (-xi,-1,1\;\xi)$.
      \item 2 \CZ s touching qubit 1 iff
        $S(D) = (a,b,c;abc), (a,b,c;ab/c), (a,b,c;ac/b)$.
    \end{itemize}
  \end{lemma}
  \begin{proof}
    This is just a translation of Theorem \ref{thm:ccount}
    using Observation \ref{obs:inv:s}, involving
    a straightforward  but tedious calculation which we omit.
  \end{proof}

  The two possibilities $S(D) = (a,b,c;abc), (a,b,c;ab/c)$ are
  quite different, and the following result helps distinguish
  between them.

  \begin{lemma} \label{lem:aud}
    Let $D$ be a three-qubit diagonal operator and $u$ be
    a one-qubit gate.  Suppose
    $|D  u^{(3)} \CZ^{(1,3)}|_{\CZ;1} = 1$ or
    $|\CZ^{(1,3)}  u^{(3)} D|_{\CZ;1} = 1$.
    Then
    $\lambda_1(D) \lambda_2(D) = \lambda_3(D) \xi(D)$.
  \end{lemma}
  \begin{proof}
    The conclusion being stable under $D \to D^\dagger$, we
    assume  $|D  u^{(3)} \CZ^{(1,3)}|_{\CZ;1} = 1$.
    Decompose $u^\dagger = e^{i \theta} R_z(\alpha) R_y(\beta) R_z(\gamma)$.
    Then
    $\inv^{(\ell)}(A)$ is given by the roots of the polynomials
    \[x^2 - \cos(2\beta) (1 - \lambda_2(D))x - \lambda_2(D) \]
    \[x^2 - \cos(2\beta) (\lambda_3(D) - \xi /(D) \lambda_1(D))x -
    \lambda_3(D) \xi(D) /
    \lambda_1(D) \]
    For these to have roots either $\{p,p,-p,-p\}$
    or $\{p,p,p,p\}$, the two equations  must
    have the same constant terms -- either both $p^2$ or both $-p^2$.
  \end{proof}

  We turn to computing \CZ -costs.  These being invariant
  under relabelling of qubits,
  we write $s(D)$ for $(\lambda_1(D), \lambda_2(D), \lambda_3(D); \xi(D))$,
  where we ignore the order of the $\lambda_i$.

  \begin{observation}
    Given two three-qubit diagonals $D, D'$,
    $s(D) = s(D')$ if and only if there exist one-qubit diagonals $d,d',d''$
    and a wire permutation $\omega$ such that
    $D = (d \otimes d' \otimes d'') \cdot \omega D \omega^\dagger$.
    Thus $s(D) = s(D') \implies |D|_{\CZ} = |D'|_{\CZ}$.
  \end{observation}

  \begin{theorem}\label{thm:3qd}
    Let $D$ be a three-qubit diagonal operator. Then there exists
    a \CZ -circuit for $D$ containing
    \begin{itemize}
    \item 0 \CZ s iff $s(D) = (1,1,1;1)$.
    \item 1 \CZ\ iff $s(D) = (1,1,-1;-1)$.
    \item 2 \CZ s iff $s(D) = (1,1,\xi;\xi), (1,-1,-1;1)$.
    \item 3 \CZ s iff  $s(D) = (1,1,\xi;\xi),
      (\xi, -1, -1; \xi), (-\xi, 1, -1; \xi)$.
    \item 4 \CZ s iff $s(D) = (a,b,c;ab/c)$.
    \item 5 \CZ s iff $s(D) = (a,b,c;ab/c), (a,b,c;abc)$
    \item 6 \CZ s always
    \end{itemize}
  \end{theorem}
  \begin{proof}
    We assume without loss of generality that $D$ takes the form
    $\mbox{diag}(1,1,1,\lambda_1,1,\lambda_2,\lambda_3,\xi)$.
    We number the qubits 1,2,3 from most to least significant.

    $(\Leftarrow)$.  We can assume that
    in fact $S(D)$ takes the form given. Our constructions
    will use the $\CX$, which may be replaced
    by the $\CZ$ at the cost of inserting $\mathtt{HADAMARD}$
    gates.

    {\bf Case 0.} $S(D) = (1,1,1;1) \implies D = I$.

    {\bf Case 1.} $S(D) = (1,1,-1;-1) \implies D = \CZ^{(1,2)}$.

    {\bf Case 2a.} $S(D) = (\xi,1,1;\xi)$.  Fix $\eta = \sqrt{\xi}$;
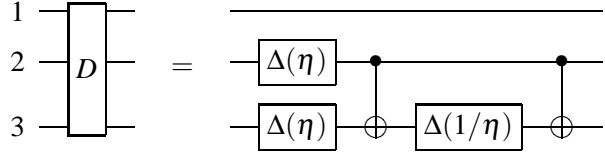
    \[
    \qc{
      \lstick{1} & \multigate{2}{D} & \qw & \ccteq{2} & &
      \qw
      & \qw & \qw & \qw &  \qw \\
      \lstick{2} & \ghost{D} & \qw & \ccteqg & &
      \gate{\Delta(\eta)}
      & \control \qw & \qw & \control \qw &  \qw \\
      \lstick{3} & \ghost{D} & \qw & \ccteqg & &
      \gate{\Delta(\eta)} & \targ \qwx &
      \gate{\Delta(1/\eta)}
      & \targ \qwx &  \qw }
    \]

    {\bf Case 2b.} $S(D) = (1,-1,-1;1) \implies D = \CZ^{(1,3)}\CZ^{(1,2)}$.

    {\bf Case 3a.} $S(D) = (\xi,1,1;\xi)$.
    By Case 2a, the \CZ\ can be implemented in a circuit
    containing 2 \CZ s.  It follows that any
    operator that can be implemented with $n > 0$ \CZ s can be implemented
    with $n+1$.  Thus since $D$ can be implemented with 2 \CZ s,
    it can be implemented with 3.

    {\bf Case 3b.} $S(D) = (\xi,-1,-1;\xi)$.  Fix $\eta = \sqrt{\xi}$;
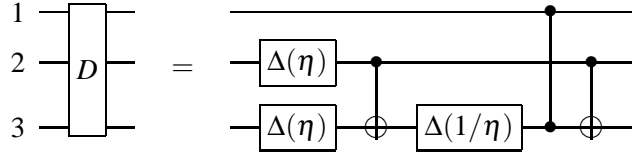
    \[
    \qc{
      \lstick{1} & \multigate{2}{D} & \qw & \ccteq{2} & &
      \qw
      & \qw & \qw & \control \qw & \qw & \qw  \\
      \lstick{2} & \ghost{D} & \qw & \ccteqg & &
      \gate{\Delta(\eta)}
      & \control \qw & \qw & \qw \qwx & \control \qw &  \qw \\
      \lstick{3} & \ghost{D} & \qw & \ccteqg & &
      \gate{\Delta(\eta)} & \targ \qwx &
      \gate{\Delta(1/\eta)}
      & \control \qw \qwx & \targ \qwx  & \qw }
    \]

    {\bf Case 3c.} $S(D) = (-\xi,1,-1;\xi)$.  Fix $\eta = \sqrt{-\xi}$.
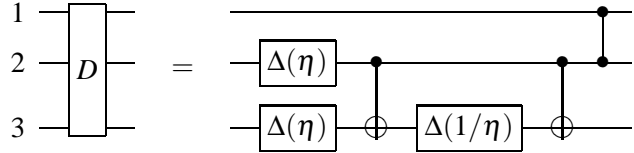
    \[
    \qc{
      \lstick{1} & \multigate{2}{D} & \qw & \ccteq{2} & &
      \qw
      & \qw & \qw & \qw & \control \qw & \qw  \\
      \lstick{2} & \ghost{D} & \qw & \ccteqg & &
      \gate{\Delta(\eta)}
      & \control \qw & \qw &  \control \qw  & \control \qw \qwx &  \qw \\
      \lstick{3} & \ghost{D} & \qw & \ccteqg & &
      \gate{\Delta(\eta)} & \targ \qwx &
      \gate{\Delta(1/\eta)}
      &\targ \qwx  & \qw & \qw }
    \]

    {\bf Case 4.} $S(D) = (a,b,c;ab/c)$.  Fix square roots
    $\alpha, \beta, \gamma$ for $a,b,c$;
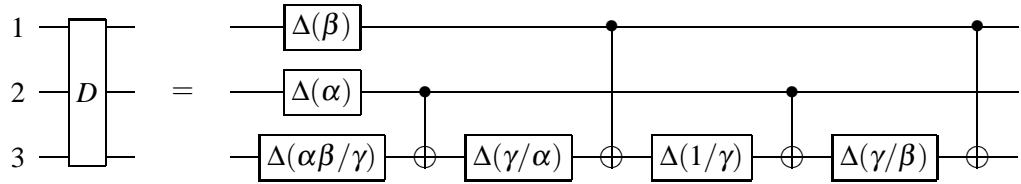
    \[
    \qc {
      \lstick{1}
      & \multigate{2}{D} & \qw & \ccteq{2} & & \gate{\Delta(\beta)}
      & \qw & \qw
      & \control \qw & \qw &
      \qw & \qw & \control \qw & \qw \\
      \lstick{2}
      & \ghost{D} & \qw & \ccteqg & & \gate{\Delta(\alpha)}
      & \control \qw & \qw &
      \qw \qwx & \qw & \control \qw & \qw & \qw \qwx & \qw \\
      \lstick{3}
      & \ghost{D} & \qw & \ccteqg & & \gate{\Delta(\alpha \beta /\gamma)}
      & \targ \qwx &
      \gate{\Delta(\gamma / \alpha)} & \targ \qwx & \gate{\Delta(1/\gamma)}
      & \targ \qwx &
      \gate{\Delta(\gamma/\beta)} & \targ \qwx & \qw
    }
    \]

    {\bf Case 5a.} $S(D) = (a,b,c;ab/c)$.  As $D$ can be implemented with
    4 \CZ s, it can be implemented with 5.

    {\bf Case 5b.} $S(D) = (a,b,c;abc)$. Fix square roots
    $\alpha, \beta, \gamma$ for $a,b,c$;
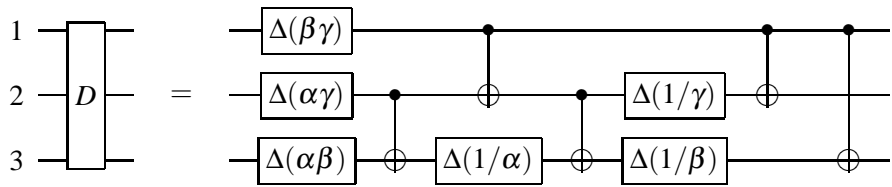
    \[
    \qc {
      \lstick{1}
      & \multigate{2}{D} & \qw & \ccteq{2} & & \gate{\Delta(\beta\gamma)}
      & \qw & \control \qw &
      \qw & \qw & \control \qw & \qw & \control \qw & \qw \\
      \lstick{2}
      & \ghost{D} & \qw & \ccteqg & & \gate{\Delta(\alpha\gamma)}
      & \control \qw  & \targ \qwx & \control \qw &
      \gate{\Delta(1/\gamma)} & \targ \qwx & \qw & \qw \qwx & \qw \\
      \lstick{3}
      & \ghost{D} & \qw & \ccteqg & & \gate{\Delta(\alpha \beta)}
      & \targ \qwx & \gate{\Delta(1 / \alpha)} &
      \targ \qwx & \gate{\Delta(1/\beta)} & \qw &
      \qw & \targ \qwx & \qw
    }
    \]

    {\bf Case 6.} More generally,
    any n-qubit diagonal operator has \CZ-cost bounded
    by $2^n - 2$.  See \cite{Bullock:diag:04} or Section
    \ref{sec:cartan}.

    $(\Rightarrow)$.

    {\bf Case 0.} $D$ must be locally equivalent to $I$, hence
    $s(D) = (1,1,1;1)$.

    {\bf Case 1.} $D$ must be locally equivalent to some \CZ , hence
    $s(D) = (1,1,-1;-1)$

    {\bf Case 2} Suppose there exists a minimal implementation of
    $D$ in which both $\CZ$ gates connect the same two qubits.  Then
    $D$ is locally equivalent to a two-qubit diagonal; in which case
    one can compute
    $s(D) = (\xi,1,1;\xi)$

    Otherwise, there is a minimal implementation of
    $D$ in which the two \CZ\ gates are $\CZ^{(i,j)}$, $\CZ^{(j,k)}$.
    By Corollary \ref{cor:twoc}, we may pass to an implementation
    with only diagonal one-qubit gates along $j$; by Corollary
    \ref{cor:sp1:replace}, we may pass to an implementation
    with only diagonal one-qubit gates along $i,k$ as well.  But then
    $D$ is locally equivalent to $\CZ^{(i,j)}\CZ^{(j,k)}$ and we may compute
    $s(D) = (1,-1,-1;1)$.

    {\bf Case 3.} It suffices to show that
    $|D|_{\CZ;j} \le 1$ for some $j$.  For, if
    $|D|_{\CZ;j} = 0$, then $D$ is a two-qubit diagonal, with
    $s(D) = (\xi, 1, 1; \xi)$, and if $|D|_{\CZ;j} = 1$, then
    by Lemma \ref{lem:3qd},
    $s(D) = (-\xi, 1, -1; \xi)$ or $(\xi, -1, -1; \xi)$.

    Consider an implementation of $D$ containing three \CZ s.
    We have $|D|_{\CZ;\ell} \le 1$
    for some $\ell$ unless the \CZ s are
    distributed so that each qubit touches exactly
    two.  Let $j$ be a qubit touching the middle \CZ.  By Corollary
    \ref{cor:twoc}, we can assume the circuit contains only diagonal
    gates on qubit $j$; it follows by inspection that
    $D \sim_j \CZ^{(i,j)} \CZ^{(j,k)}$.  But we have already
    determined that $|\CZ^{(i,j)} \CZ^{(j,k)}|_{\CZ;j} = 1$.

    {\bf Case 4.} Consider an implementation of $D$ containing
    four \CZ s.  If any qubit touches fewer than two \CZ s,
    we reduce to the previous case and observe that the desired
    condition on $s$ holds.  Thus suppose each qubit touches at
    least two \CZ s.  Then there are only two possibilities for the
    number of \CZ s touched by each qubit:  $(2,2,4)$ and
    $(2,3,3)$.

    For the configuration $(2,2,4)$, say qubits $\ell,m$ touch two \CZ s
    and qubit $n$ touches four.  Note that no \CZ s connect $\ell, m$.
    Thus we may assume by Corollary
    \ref{cor:twoc} all one-qubit gates on $\ell, m$ are diagonal.
    By Proposition \ref{prop:noninvariance}, $\det_{\ell,m} D$ is
    separable; this says precisely that $\lambda_{\ell}(D) \lambda_{m}(D) = \lambda_n(D) \xi(D)$.

    For the configuration $(2,3,3)$, say qubit 1 touches two \CZ s
    and qubits 2,3 touch three.  Then there are two \CZ s connecting
    qubits 2 and 3,
    one connecting qubits 1 and 3 and one connecting qubits 1 and 2.
    By Corollary \ref{cor:twoc}, we ensure that all one-qubit gates on
    qubit 1 are diagonal.
    If the \CZ s connecting qubits 2 and 3 are outermost,
    $D \sim_\ell \CZ^{(1,2)} \CZ^{(1,3)}$,
    hence can be implemented with three $\CZ$ s by Case 3.
    Otherwise, one of the \CZ s incident on qubit 1 is outermost;
    without loss of generality let it be $\CZ^{(1,3)}$.  Then
    we have an equation of the form $D = u^{(3)} \CZ^{(1,3)} A$ where
    by construction $A$ commutes with $\Z^{(1)}$ and $|A|_{\CZ;1} = 1$.
    Lemma \ref{lem:aud} yields the desired result.

    {\bf Case 5.} It suffices by Lemma \ref{lem:3qd}
    to show that $|D|_{\CZ;\ell} \le 2$ for
    some $\ell$.  Suppose not; then in any five-\CZ\ implementation
    for $D$, each qubit must touch three \CZ s.
    It follows that two of the qubits, say $\ell, m$
    touch exactly three \CZ s, and the remaining qubit
    touches four.  By Theorem \ref{thm:threec},
    all one-qubit gates on $\ell, m$ are diagonal or anti-diagonal.
    Enough applications
    of Equation \ref{circ:demorgan} will ensure that all one-qubit
    gates on $\ell, m$ are in fact diagonal.
    Move the \CZ\ which connects $\ell, m$
    to the edge of the circuit.  This yields
    $D = \CZ^{(\ell,m)}A$, where
    $|A|_{\CZ;\ell} \le 2$.  By Lemma
    \ref{lem:3qd}, it follows that $|D|_{\CZ;\ell} \le 2$ as well.
  \end{proof}

  \section{Circuits with ancillae} \label{sec:ancilla}

  The proofs of Theorems \ref{thm:T6} and \ref{thm:3qd}
  assume that only three qubits were present, and use this assumption
  when enumerating possible circuit configurations with
  a given total number of \CZ\ gates.  This dependency
  can be eliminated. Indeed, these cases involved so few \CZ s
  that one could eliminate configurations with ancillae
  by performing explicit checks.

  More significant is the use of
  Proposition \ref{prop:noninvariance}
  and the characterization by Theorem \ref{thm:ccount}
  of $|D|_{\CZ;\ell} \le 2$.  Both of these statements
  are true for any fixed $N$, but suffer when $N$ is allowed
  to vary.  For example
  if only $N = 3$ qubits are available, then
  $\det_{1,2}\CCZ^{(1,2,3)} = \CZ^{(1,2)}$, so
  by Proposition \ref{prop:noninvariance},
  the \CCZ\ cannot be implemented
  in any three-qubit \CZ -circuit in which all gates
  commute with $\Z^{(1)}, \Z^{(2)}$.  But if $N = 4$ qubits are
  present,
  $\det_{1,2}(\CCZ^{(1,2,3)}) = I^{(1,2)}$,
  so $\CCZ^{(1,2,3)}\otimes I^{(4)}$ {\em can} be implemented in
  a four-qubit \CZ -circuit in which all one-qubit
  gates commute with $\Z^{(1)}$ and $\Z^{(2)}$.

  Similarly, for $N = 3$ qubits, we have
  $\inv^{(\ell)}(\CCZ) = \{1,1,1,-1\}$
  and thus by Theorem \ref{thm:ccount}
  $|\CCZ|_{\CZ;\ell} \ge 3$.  However,
  for $N = 4$ qubits,
  $\inv^{\ell}(\CCZ^{(\ell,m,n)}) = \{1,1,1,-1,1,1,1,-1\}$,
  so now Theorem \ref{thm:ccount} implies that
  $|\CCZ^{(1,2,3)}\otimes I^{(4)}|_{\CZ;1} = 2$.  Indeed:
  \[
  \qc{
    & \control \qw & \qw & \ccteq{3} & & \qw & \control \qw & \qw & \qw
    & \qw & \control \qw & \qw & \qw & \qw \\
    & \qw \qwx & \qw & \ccteqg & & \gate{H} & \control \qw \qwx
    & \gate{H} & \control \qw & \gate{H} & \control \qw \qwx &
    \gate{H} & \control \qw & \qw \\
    & \control \qw \qwx & \qw & \ccteqg & & \qw & \qw & \qw &
    \control \qw \qwx &  \qw & \qw & \qw &
    \control \qw \qwx & \qw \\
    & \control \qw \qwx & \qw & \ccteqg & & \qw & \qw & \qw &
    \control \qw \qwx &  \qw & \qw & \qw &
    \control \qw \qwx & \qw
  }
  \]

  On the other hand,
  the properties $\inv^{(\ell)}(U) \cong \{1,1,\ldots\}$ and
  $\inv^{(\ell)}(U) \cong \{1,-1,1,-1\ldots\}$ are stable
  under adding ancillae.
  By Theorem \ref{thm:ccount},
  so are the properties $|U|_{\CZ;\ell} = 0$ and
  $|U|_{\CZ;\ell} = 1$.  Since only these properties
  are used in the proof of Lemma \ref{lem:aud}, it
  too holds even in the presence of ancillae.
  This leads to an extension of the \CZ -cost classification
  of three-qubit diagonals
  to the case where ancilla qubits are permitted.

  \begin{lemma} \label{lem:czanc}
    Let $A$ be a unitary operator;
    let $\mathcal{C}$ be qubit minimal among \CZ-circuits
    computing $A$, possibly with the use of ancillae,
    using only $|A|^a_{\CZ}$ \CZ\ gates.
    Then every ancilla in
    $\mathcal{C}$ touches at least three \CZ\ gates.
  \end{lemma}
  \begin{proof}
    Fix an ancilla qubit $\ell$.  If no \CZ\ gates touch $\ell$,
    then it may be removed.  If one (respectively two)
    \CZ\ touches $\ell$, then by Corollary \ref{cor:sp1:replace}
    (respectively Corollary \ref{cor:twoc}),
    then there is a circuit with no more
    \CZ s in which the only one-qubit gates on $a$ are
    diagonal.

    Now form the circuit $\bra{0}^{(\ell)} \mathcal{C} \ket{0}^{(\ell)}$ as
    in the proof of Corollary \ref{cor:multicz}.  This circuit
    computes the operator $A$ using one fewer ancilla,
    fewer \CZ s than $\mathcal{C}$.
  \end{proof}

  \begin{corollary}
    For any two-qubit operator $V$,
    $|V|^a_{\CZ} = |V|_{\CZ}$.
  \end{corollary}
  \begin{proof}
    If no ancillae are needed to minimize \CZ-count, then the result
    holds.  Otherwise, each ancilla used in a qubit-minimal
    \CZ-minimal implementation must touch at least three \CZ
    gates. Thus $|\cdot|_{\CZ} \ge |\cdot|^a_{\CZ} \ge 3$.
    However it is known \cite{Vidal:2q:04, Vatan:2q:04, Shende:2q:04} that
    two-qubit operators have $|\cdot|_{\CZ} \le 3$.  Thus all the
    inequalities are equalities.
  \end{proof}

  \begin{proposition}
    For any three-qubit diagonal operator, $D$,
    $|D|^a_{\CZ} = |D|_\CZ$.
  \end{proposition}
  \begin{proof}
    Suppose $|D|^a_{\CZ} < |D|_\CZ$.  By Lemma \ref{lem:czanc},
    a qubit-minimal circuit for $D$ achieving the
    bound for $|D|^a_{\CZ}$ contains at least three
    \CZ\ gates incident on each ancilla.  By assumption
    at least one ancilla is used, so
    $|D|_\CZ > |D|^a_{\CZ} \ge 3$.  It follows
    from Theorem \ref{thm:3qd} and Lemma \ref{lem:3qd} that
    $|D|_{\CZ;\ell} > 1$ for the three qubits $\ell = 1,2,3$.
    By Theorem \ref{thm:ccount}, this property is
    stable under addition of ancilla.  Thus
    a qubit-minimal circuit for $D$ achieving the
    bound for $|D|^a_{\CZ}$ contains at least 3 \CZ s
    incident to each ancilla, and at least 2 \CZ s incident
    to each non-ancilla qubit.  If $k$ ancillae
    are used, then we have
    $|D|^a_{\CZ} \ge (3k + 6)/2$.  From Theorem
    \ref{thm:3qd} and the supposition we have
    $|D|^a_{\CZ} < |D|_\CZ = 6$; it follows that $k=1$,
    that $|D|^a_{\CZ} = 5$, and that $|D|_\CZ = 6$.

    In any four-qubit, five-\CZ\ circuit for $D$, we must have
    two of the non-ancilla, say $x_1, x_2$ touching two \CZ s, and
    both the remaining non-ancilla $z$ and the ancilla $a$ touching
    three.  By Corollary \ref{cor:twoc}, we can assume that the
    only one-qubit operators appearing on $x_1$, $x_2$ are diagonal.
    We may also assume that the graph where vertices
    are qubits and edges are \CZ\ gates is connected; otherwise
    $D$ could be split into the tensor product of a two-qubit
    and a one-qubit diagonal, and hence would have $|D|\le 2$.
    Then there are only three
    possibilities regarding which wires are connected by \CZ s.

    \[
    \begin{array}{l|ccccc}
      I & (x_1, x_2) & (x_1, z)   & (x_2, a) & (z,a) & (z,a) \\
      \hline
      II & (x_1, z)   & (x_1, z)   & (x_2, a) & (x_2, a) & (z,a) \\
      \hline
      III & (x_1, z)   & (x_1, a)   & (x_2, z) & (x_2, a) & (z,a)
    \end{array}
    \]

    We will show that any circuit with those \CZ\ gates
    can be transformed so that (*) a \CZ\
    which does not touch the ancilla is outermost among the \CZ s,
    and (**) one of the $x$-qubits on which this \CZ\ gate
    acts has the property that all one-qubit gates acting on it are diagonal.
    As this $x$-qubit only touched 2 \CZ\ gates to begin with, it
    follows from Lemma \ref{lem:aud} that $s(D)$ takes the form
    $(a,b,c;ab/c)$. By Theorem \ref{thm:3qd},
    $|D|_\CZ = 4$, which is a contradiction.

    We return to checking (*) and (**).
    Eliminate non-diagonal one-qubit gates on $x_i$
    using Corollary \ref{cor:twoc}.  In Case (I), the
    $(x_1,x_2)$ \CZ\ can therefore only be prevented from moving
    by the $(x_1,a)$.  This can be on only one side, so the $(x_1, x_2)$
    can be moved outwards to the other.  Similarly, in Case (II),
    an $(x,z)$ can only be blocked by $(z,a)$ and the other
    $(x,z)$.  In this case, the second $(x,z)$ is blocked on
    only one side and can be moved to the edge.  In Case (III),
    we use Corollary \ref{cor:twoc} to clear both the $x_1$ and $x_2$
    qubits of non-diagonal gates; the possible additional
    one-qubit gates will only fall on the $z$ and $a$ qubits.
    Now the $(x_1,z)$
    can only be blocked by the $(x_2,z)$ and the $(z,a)$, and
    also the $(x_2,z)$ can only be blocked by $(z,a)$ and $(x_1,z)$.
    Thus one of $(x_1,z)$ and $(x_2,z)$ can be made outermost.
  \end{proof}

  \begin{corollary}
    $|\CCZ|^a_{\CZ} =  |\TOFFOLI|^a_{\CZ} = 6$ and $|PERES|^a_{\CZ} = 5$.
  \end{corollary}

  \section{Conclusion}
  \label{sec:conclusion}

  While our work is primarily focused on quantum circuit implementations,
  the \TOFFOLI\ gate originally arose as a universal gate for classical
  reversible logic \cite{Toffoli:80}.
  In contrast, the \NOT\ and \CNOT\ gates are {\em not} universal for reversible
  logic:  their action on bit-strings is affine-linear over
  over $\FF_2$, and thus the same is true for any operator
  computed by any circuit containing only these gates.

  Augmenting \CNOT\ gates with single-qubit rotations
  to express the \TOFFOLI\ gate provides the lacking non-linearity.
  Thus the number of one-qubit gates (excluding inverters) needed
  to express the \TOFFOLI , or more generally any reversible computation,
  can be thought of as a measure of its non-linearity.
  In this inverted cost model (also relevant to some quantum implementation
  technologies)  the following question remains open: {\em how many one-qubit gates
  are needed to implement the \TOFFOLI ?}
  Furthermore,
  {\em are there circuits that simultaneously minimize
  the number of \CNOT\ and one-qubit gates ?}

  In a different direction, recall our results showing
  that diagonality and block-diagonality of an operator
  impose strong constraints on small circuits
  that compute this operator. We believe other
  conditions may act in a similar way.
  In particular, we ask {\em what can be said about
  minimal quantum circuits for operators computable
  by classical reversible circuits, i.e., operators expressed by 0-1 matrices?}
  Very little is known even for three-qubit operators.
  In particular, the \CNOT -cost of the controlled-swap
  (Fredkin gate) remains unresolved.

  Closest to our present work, the exact \CNOT -cost of the $n$-qubit
  analogue of the \TOFFOLI\ gate remains unknown.  We have
  shown that $2n$ \CNOT s are necessary if ancillae are not
  permitted, but already for $n=4$ we only know that
  $8 \le |{\tt CCCZ}|_{\CZ} \le 14$, where the upper bound is provided by
  a generic decomposition of diagonal operators \cite{Bullock:diag:04}.
  Existing constructions of the $n$-qubit \TOFFOLI\ gate require
  a quadratic number of \CNOT\ gates  without the use of ancillae.
  With one ancilla, such constructions require linearly many \CNOT s,
  but the leading coefficient is in double-digits
  \cite{Barenco:elementary:95,Maslov:03}.

  Finally, we hope that our proof can be simplified
  and our techniques generalized. In particular, we have relied
  on repeated comparisons of various Cartan  decompositions
  to each other.
  A careful study of the proof will reveal the simultaneous use
  of six Cartan decompositions --- those corresponding to conjugation
  by $\X$ and $\Z$ on each of three wires. Keeping track of these
  decompositions in a more systematic manner may simplify the proof,
  while using additional decompositions may lead to new results.
  A related challenge is gauging the power of the qubit-by-qubit
  gate counting we have used.  It follows from the results of \cite{Shende:qsd:05}
  that $|U|_{\CZ;\ell} < 6(n-1)$ for $U$ an $n$-qubit operator, and hence
  no technique relying solely on this process can achieve better than
  a quadratic lower bound. On the other hand, we have only been able
  to characterize cases when $|U|_{\CZ;\ell} > 2$, and thus
  have achieved only linear lower bounds.

  \vspace{2mm} \noindent {\bf Acknowledgements.}
  We thank Mikko Mottonen, Jun Zhang, K. Birgitta Whaley, and
  Yaoyun Shi for helpful discussions.
  This work was sponsored in part by
  the Air Force Research Laboratory
  under Agreement No. FA8750-05-1-0282.

    \section*{Appendix: Proof of Proposition
  \ref{prop:noninvariance}}

   Below we restate Proposition \ref{prop:noninvariance} and complete its proof.
  \setcounter{theorem}{5}
  \addtocounter{theorem}{-1}
  \begin{proposition}
    Fix qubits $\ell_1 \ldots \ell_k$ among $N > k$ qubits.
    A unitary $U$ commuting with $\Z^{(\ell_1)}, \ldots,
    \Z^{(\ell_k)}$
    can be implemented by a \CZ-circuit in which only diagonal
    gates operate on qubits $\ell_i$ if and only if
    $\det_{\ell_1 \ldots \ell_k}(U)$
    is separable (can be implemented by one-qubit gates).
  \end{proposition}
  \begin{proof}
    ($\Rightarrow$).
    It suffices to show the separability of
    $\det_{\ell_1 \ldots \ell_k}(U)$ for a small generating set
    of operators.
    Direct calculation confirms this for (i) \CZ\ gates,
    (ii) diagonal one-qubit gates on the $\ell_i$, and
    (iii) any gate not affecting qubits $\ell_i$.

    ($\Leftarrow$).
    By hypothesis, $\det_{\ell_1 \ldots \ell_k}(U)$, and hence
    ${\mathcal D}=\det_{\ell_1 \ldots \ell_k}(U)^{-2^{k-N}}$,
    can be implemented
    using only one-qubit diagonal gates.  It remains
    to implement
    $\tilde{U} = U / {\mathcal D}$, which satisfies
    the normalization
    $\tilde{U}_{j_1 \ldots j_k} \in \SUnitary(2^{N-k})$.
    We will construct a circuit for $\tilde{U}$ by multiplexing
    circuits for $\tilde{U}_{j_1 \ldots j_k}$.
    Let $\mathcal{C}$ be a $(N-k)$-qubit circuit
    containing only \CZ s and one-qubit $R_x, R_y, R_z$ gates
    such that any operator in $\SUnitary(2^{N-k})$
    can be implemented by making the appropriate choice of
    parameter for the $R_x, R_y, R_z$ gates.
    Such universal circuits
    exist \cite{Barenco:elementary:95};
    see Section \ref{sec:cartan}
    for modern constructions.
    Choose specifications $\mathcal{C}_{j_1 \ldots j_k}$ implementing the
    $\tilde{U}_{j_i \ldots j_k}$; let the $s$-th rotation
    gate in $\mathcal{C}_{j_1 \ldots j_k}$ be given by
    $R_{d(s)}(\theta_{j_1 \ldots j_k}(s))^{(q(s))}$, where
    $q(s)$ is a qubit,
    $\theta_{j_1 \ldots j_k}(s)$ is an angle, and $d(s) = x, y, z$.
    Define $\Theta(s)$
    to be the real diagonal operator on qubits $\ell_i \ldots \ell_k$
    such that $\Theta(s)_{j_i \ldots j_k} = \theta_{j_1 \ldots j_k}(s)$.
    Form the $N$-qubit circuit $\tilde{\mathcal{C}}$ by replacing
    the $s$-th rotation gate of
    $\mathcal{C}$ by
    the multiplexed rotation $R_{d(s)}(\Theta(s))^{(q(s))}$;
    then $\tilde{\mathcal{C}}$ implements $\tilde{U}$.
    Implement $R_{d(s)}(\Theta(s))^{(q(s))}$ by
    a \CZ-circuit containing
    no one-qubit operator on any qubit save
    $q(s)$, which is not one of the $\ell_i$ (see
    \cite{Mottonen:csd:04} or Section \ref{sec:cartan}).
  \end{proof}

  \setcounter{theorem}{34}
  \begin{corollary}
    $N$-qubit operators
    which commute with $\Z$ on $k$ qubits can be implemented using on
    the order of $2^k 4^{N-k}$ one-qubit and \CZ\ gates.\footnote{
    Dimension-counting following \cite{Knill:95} shows that roughly
    this many are necessary for almost all such operators.}
  \end{corollary}
  \begin{proof}
    This follows from the construction in the proof of Proposition
    \ref{prop:noninvariance} and the known estimates
    in the cases $k = 0, N-1$ \cite{Mottonen:csd:04}
    and $k=N$ \cite{Bullock:diag:04}.
  \end{proof}

\end{document}